\newif\ifquantumjournal
\quantumjournaltrue

\newif\ifsodasubmission
\sodasubmissionfalse

\ifquantumjournal
	\documentclass[a4paper,11pt,onecolumn]{quantumarticle}
	\pdfoutput = 1
\else
	\documentclass[11pt]{scrartcl}
\fi

\ifsodasubmission
	\usepackage[letterpaper, margin=1in]{geometry}
\else
	\usepackage[a4paper, left=2.5cm, right=2.5cm, top=2.5cm, bottom=2.5cm]{geometry}
\fi

\usepackage[utf8]{inputenc}
\usepackage{etex}

\pdfoutput=1

\usepackage[english]{babel}
\usepackage{wrapfig}

\usepackage{amsthm}
\usepackage{hyperref}
\usepackage{color}
\usepackage{mathtools}
\usepackage{amssymb}

\usepackage{amsfonts}

\usepackage{authblk}
\ifquantumjournal

\else
	
\fi

\usepackage{physics}
\usepackage{float}
\usepackage{enumitem}
\setlist{noitemsep} \usepackage{algcompatible}
\usepackage{subfig}
\usepackage{bm} 
\usepackage{qcircuit}
\usepackage{adjustbox}
\usepackage{multirow}
\usepackage[table,xcdraw]{xcolor}

\usepackage{tikz, graphics} \usetikzlibrary{arrows.meta, bending,     patterns,     shapes,positioning
               }

\usepackage [
  n,
  advantage,
  operators,
  sets,
  adversary,
  landau,
  probability, 
  notions,
  logic,
  ff,
  mm,
  primitives,
  events,
  complexity,
  asymptotics,
  keys
  ] {cryptocode}

\usepackage[strings]{underscore}
\usepackage{bm}
\usepackage{booktabs} 

\renewcommand{\cc}[1]{\complclass{#1}}
\newcommand{\Abort}{\mathsf{Abort}}

\newcommand{\BQP}{\mathsf{BQP}}

\newcommand{\X}{\mathsf{X}}
\newcommand{\Y}{\mathsf{Y}}
\newcommand{\Z}{\mathsf{Z}}
\newcommand{\Id}{\mathsf{I}}
\newcommand{\Ha}{\mathsf{H}}
\newcommand{\CZ}{\mathsf{CZ}}

\newcommand{\CNOT}{\mathsf{CNOT}}

\newcommand{\cptp}[1]{\mathsf{#1}}

\floatstyle{ruled}

\newfloat{resource}{htb!}{idf}
\floatname{resource}{Resource}

\newfloat{simulator}{htb!}{idf}
\floatname{simulator}{Simulator}

\newfloat{protocol}{htb!}{idf}
\floatname{protocol}{Protocol}

\newfloat{routine}{htb!}{idf}
\floatname{routine}{Routine}

\usepackage{todonotes}
\todostyle{holl}{color=green!30, size=\footnotesize}
\todostyle{holli}{color=green!30, inline, size=\footnotesize}
\todostyle{dl}{color=blue!30, size=\footnotesize}
\todostyle{dli}{color=blue!30, inline, size=\footnotesize}
\todostyle{lm}{color=red!30, size=\footnotesize}
\todostyle{lmi}{color=red!30, inline, size=\footnotesize}
\todostyle{tk}{color=yellow!30, size=\footnotesize}
\todostyle{tki}{color=yellow!30, inline, size=\footnotesize}
\theoremstyle{plain}
\newtheorem{theorem}{Theorem}[section]
\newtheorem{lemma}[theorem]{Lemma}
\newtheorem{corollary}[theorem]{Corollary}

\theoremstyle{definition}
\newtheorem{definition}[theorem]{Definition}
\newtheorem{resourceenv}{Resource}
\newtheorem{protocolenv}{Protocol}
\newtheorem{simulatorenv}{Simulator}
 
\begin{document}

\title{Plugging Leaks in Fault-Tolerant Quantum Computation and Verification}
\date{}

\ifsodasubmission
	\author{}
\else
	\ifquantumjournal
		\author[1]{Theodoros Kapourniotis}
		\thanks{This work was partially developed while at the University of Warwick.}
	\else
		\author[1]{Theodoros Kapourniotis\thanks{This work was partially developed while at the University of Warwick.}}
	\fi
	\author[2]{Dominik Leichtle}
	\author[3]{Luka Music}
	\author[4]{Harold Ollivier}
	\affil[1]{National Quantum Computing Centre, Rutherford Appleton Laboratory, Harwell Campus, Didcot, Oxfordshire, OX11 0QX, United Kingdom}
	\affil[2]{School of Informatics, University of Edinburgh, 10 Crichton Street, Edinburgh EH8 9AB, United Kingdom}
	\affil[3]{Quandela, 7 Rue Léonard de Vinci, 91300 Massy, France}
	\affil[4]{Département d'Informatique, Ecole Normale Supérieure, 45 rue d'Ulm, 75005 Paris, France (Université PSL - CNRS - INRIA)}
\fi

\maketitle

\begin{abstract}
  With the advent of quantum cloud computing, the security of delegated quantum computation has become of utmost importance. While multiple statistically secure blind verification schemes in the prepare-and-send model have been proposed, none of them achieves full quantum fault-tolerance, a prerequisite for \emph{useful} verification on scalable quantum computers. In this paper, we present the first fault-tolerant blind verification scheme for universal quantum computations able to handle secret-dependent noise on the verifier's quantum device. Composable security of the proposed protocol is proven in the Abstract Cryptography framework.
  
  Our main tools are two novel distillation protocols that turn secret-dependent noise into secret-independent noise. The first one is run by the verifier and acts on its noisy gates, while the second and more complex one is run entirely on the prover's device and acts on states provided by the verifier. Both are required to overcome the leakage induced by secret-dependent noise. We use these protocols to prepare states in the $\X - \Y$-plane whose noise is overwhelmingly secret-independent, which then allows us to verify with exponential confidence arbitrary fault-tolerant $\BQP$ computations.
\end{abstract}

\newpage

\tableofcontents

\newpage

\section{Introduction}
\label{sec:intro}
The question of secure delegated quantum computation (SDQC) is a long standing research topic among quantum computer scientists
which attracts interest from several different perspectives.
Foundationally, the problem of quantum verification touches on the falsifiability of quantum mechanics, asking whether solutions to problems that can be found efficiently by quantum devices must also be efficiently verifiable.
In the language of complexity theory, this is asking the question for the power of the class \cc{QPIP} of problems admitting Quantum-Prover Interactive Proof systems.
Finally, the functionality of SDQC is of utmost practical importance, enabling clients of quantum computing services to guarantee the confidentiality and integrity of the delegated computations, making it a cornerstone task in the process of making quantum computing largely available to end-users.
In this way, it presents a fundamental building block in the quest of building the necessary trust alongside the quantum computation value-chain.

While solutions have been provided to securely delegate bounded-error quantum-polynomial-time (\cc{BQP}) computations, none has been provided yet with robustness against noise on the Verifier's side, information-theoretic security, efficiency, and scalability.

\subsection{The problem of robustness}
In spite of its practical motivation, this question has not led to solutions that could, at least in principle, be put to work with real-world systems.
The reason is that
most proposed protocols
implicitly assume that the Verifier is noise-free.
As quantum devices are inherently of an analogue nature, none of the current technology can live up to the standard of providing an entirely noise-free environment, nor can we expect to ever get sufficiently close to such a utopic situation for the assumption of noiselessness to be an accurate enough model of real-world implementations.
Consequently, any security analysis that assumes the absence of noise in any participating party opens up the possibility of side-channel attacks that take advantage of the shortcomings of this model as a description of real-world devices.

To address this issue, we are interested in \emph{robust} Secure Delegated Quantum Computation (SDQC) protocols that not only offer security against arbitrary adversaries, but also allow honest participants to perform the computation correctly with high probability in spite of implementation imperfections.
A first realization is that it is not possible to construct such a protocol with robustness against arbitrary kinds and levels of imperfections.
This is because such protocols are precisely designed to abort if the risk of obtaining an incorrect outcome is too high.
Therefore, highly correlated defects or very noisy operations will necessarily harm the success probability even of honest participants.
Said differently, very strong noise must inevitably resemble malicious behavior, and -- by secure protocols -- be recognized as such.
Our goal is thus to design protocols that are robust to a large number of well-characterized and physically relevant imperfections.
More precisely, we require that
\begin{enumerate}
\item if all parties are honest, the protocol performs the correct computation and the Verifier accepts even if imperfection affects all parties involved, up to a reasonable amount;
\item if the Prover is arbitrarily malicious, the protocol remains secure even if the Verifier's operations are imperfect, up to a reasonable amount.
\end{enumerate}

The challenge to overcome in the design of such protocols has been analyzed in~\cite{ABEM17interactive}.
In short, the only known way to construct robustness to noise in quantum computation is to perform the computation on logically encoded qubits in a fault-tolerant way.
If the Verifier only has single-qubit operations, all error correction needs to be performed by the Prover.
At the same time, to provide blindness and verifiability, the Verifier will send encrypted \emph{physical} qubits that appear totally mixed to a Prover ignorant of the encryption key.
To perform error-correction, the Verifier will thus needs to help the Prover albeit making sure that, in this process, no information about the encryption key is leaked so that security is not compromised.

\cite{ABEM17interactive} goes further by recognizing the difficulty of providing robust verification protocols even when allowing the Verifier to perform single \emph{logical} qubit operations. Indeed, using an error correcting code enlarges the attack surface for a malicious party by creating side-channels that can carry away crucial information about the secret key used by the Verifier. It is therefore necessary to precisely model how imperfections affect the devices of both the Prover and the Verifier to ensure that no such leaks occur.

\subsection{Modelling imperfect operations}
Imperfections on the Prover's devices are already taken into consideration in the security of SDQC protocols: since they are secure for arbitrary malicious Provers, they must also be secure against noisy Provers. On the other hand, imperfect Verifier devices may induce new opportunities for attacks which are not accounted for in existing protocols.

Two types of imperfections on the Verifier's device can be distinguished: (i) faults that could have equally happened on the Prover's device, and that are thus not a threat to security, but only to correctness, \emph{i.e.}, the success probability of the protocol; and (ii) leaks of parts of the Verifier's secret key, encompassing all errors and noise that are secret-dependent. The latter represents those imperfections that are dangerous to the security of SDQC protocols.
Unfortunately, our analysis reveals that standard fault-tolerant techniques applied in the context of these protocols intertwine correctness and security, making it impossible to separate faults from leaks. This fact was already noted in~\cite{ABEM17interactive}, but left as an open difficulty on the road to verification of fault-tolerant delegated quantum computations.

As a step towards the solution, we introduce \emph{Stochastically Compromised Preparations, Gates, and Measurements} which capture arbitrary combinations of stochastic leaks and faults in the imperfect building blocks of our protocols. This is done by defining physical quantum operations as parametrized sets of CPTP maps acting on a subset of physical qubits. When used, these operations can be compromised with a certain probability. This allows a malicious Prover to learn not only the set of parameters but also the specific value chosen by the Verifier, giving them the full definition of the CPTP map. The Prover can then replace the operation with any other quantum operation on the same qubits and a private register (see Resource~\ref{res:compromised-pgm-informal}). The interest of this definition is to generalize the concept of stochastic faults in fault-tolerance. When the parameters for all quantum operations are singletons, compromised operations are equivalent to faulty ones. In this case, the Prover knows the entire circuit, and once compromised locations are identified, they can apply arbitrary operations to the underlying registers.

\begin{resourceenv}{Stochastically Compromised Preparation, Gate, or Measurement (Informal)}\label{res:compromised-pgm-informal}
  \begin{algorithmic}[0]
  	\STATE \textbf{Parameters:} compromise probability $p_c$, set of CPTP maps $\{\cptp U(\lambda)\}_{\lambda \in \Lambda}$ parametrised by $\lambda \in \Lambda$.
    \STATE \textbf{User's Input:} parameter $\lambda \in \Lambda$, the quantum registers that are acted upon by the resource.
    \STATE \textbf{Eavesdropper's Input:} Flag $c \in \{0, 1\}$ indicating that the preparation, gate, or measurement is compromised.
    \STATE \textbf{Computation by the Resource:}
    \begin{itemize}
		\item Apply $\cptp U(\lambda)$ to $\rho$.
		\item If $c =1$, with probability $p_c$:
        \begin{itemize}
          \item Send $\lambda$ and the contents of the quantum register to the Eavesdropper;
          \item Receive from the Eavesdropper a quantum system and insert into its internal register.
        \end{itemize}
		\item Outputs the state in its internal register at the User's output interface.
    \end{itemize}
  \end{algorithmic}
\end{resourceenv}

Both the leaky-but-correct case and standard fault-tolerant computations can be recovered from the imperfections modelled by Resource~\ref{res:compromised-pgm-informal}.
Stochastically Leaky Remote State Preparations (see Resource~\ref{res:stoca}) correspond to Stochastic Compromised Preparation with a non-trivial parametrization set $\Lambda$ where no error is applied by the Eavesdropper.
The standard model of fault-tolerance can be obtained by assuming that the Eavesdropper always cheats. In this case there is no parametrization and the authorized leak $\Lambda$ reveals the full computation.

However, this Resource does not encompass all possible imperfections. In particular, we impose binary behavior for the leak: it either leaks the full description of the gate or nothing. We will see in Section~\ref{sec:plugging-rsp} that bounding away from zero the number of instances which do not leak, and therefore have perfect security, is a crucial requirement for the security of RSP. We leave as an open question how to deal with coherently leaky resources that always leaks at least a small amount of information about the chosen parameters $\lambda$.

\subsection{Main results}
We first consider that the noise does not need to be mitigated since it is kept at a per-qubit level that is low-enough for the computation to succeed with a sufficiently high probability. Then it is possible to take care of leaks via a protocol that is entirely performed by the Prover. The Verifier only needs to send $k$ potentially leaky qubits in the $\X-\Y$ plane for each vertex of the graph, and our gadget then exponentially suppresses the probability that the state of the qubit used in the vertex leaks. This is done by recombining the $k$ qubits into a single qubit in the , with the guarantee that the state of the final qubit is not leaked so long as at least one of the $k$ qubits did not leak. We combine this with an SDQC protocol which requires only $\X-\Y$ plane qubits and thus obtain an SDQC protocol which is resistant to side-channel attacks and leaks, in the low-noise regime. Importantly, in this case, the Verifier has no operational overhead compared to the protocol without leak-protection and the communication complexity only increases logarithmically with the final desired leak probability.

\begin{theorem}[Leak-free SDQC from Leaky Operations]
	A Secure Delegated Quantum Computation protocol can be constructed in the Abstract Cryptography framework from single-qubit Verifier operations which fully leak their secrets with probability $p_l$. The operations of the Verifier are the same as in the non-leaky protocol and the multiplicative communication overhead is logarithmic in the security parameter.
\end{theorem}

Our second main result is a threshold theorem for the secure delegation of fault-tolerant quantum computation.
In essence, it states that there exists a positive threshold for the local corruption probability of operations (which are modelled as Resource~\ref{res:scpg}), below which imperfect operations can be used to verify quantum computations of arbitrary size, while keeping both the accumulation of faults as well as the accumulation of leaks under control.
The proposed protocol is efficient in achieving exponentially low soundness errors, as well as scalable in the size of the target computation.
 
\begin{theorem}[Threshold Theorem for Secure Delegation of Fault-Tolerant Quantum Computation]
	There exists a universal threshold $p_0 > 0$ such that Secure Delegated Quantum Computation can be constructed in the Abstract Cryptography framework from imperfect operations with corruption probability below $p_0$ up to at most negligible security error.
\end{theorem}

One such concrete construction is given by Protocol~\ref{proto:ft-sdqc}. The protocol's overhead is scaling logarithmically both in the target construction error, as well as in the size of the target computation.

\subsection{Conceptual contributions}
Our contribution to the problem of verification of quantum computation in the presence of noise starts with a full analysis of how security is affected by the presence of noise. More precisely, we identify two mechanisms through which secure protocols in the noiseless case become insecure. The first one is \emph{direct}: the noise on the Verifier's side will generally depend on the secret key used to encrypt the computation. This is because these secrets are used to perform classically controlled quantum operations for which the noise will depend on the control. For instance, most physical platforms will have different noise strength for $\Id$, $\Z$, $\cptp S$ or $\cptp T$ gates. So that if the keys are used to perform a classically controlled $\Z(\theta)$ rotation, as it is the case in many information theoretically secure protocols, the noise will effectively depend on the secret key thereby voiding the validity of current security proofs. The second one is \emph{indirect}: the noise on the Verifier's side forces us to perform some encoding into logical qubits in order to provide scalability. Yet, by doing so, the syndrome qubits are able to capture partial information about the operations that are being performed. The technical reason is that applying a transversal gate leaves the syndrome qubits invariant \emph{only} when they are in the all $\ket 0$ state. Otherwise, they are generally mapped to a state that depends on the operation that is being applied, thereby compromising secrecy as the whole quantum register is eventually accessible to the Prover who is tasked with performing the computation.

Using this analysis, we propose a model for the imperfections that is compatible with the stochastic fault model used to obtain standard fault-tolerance thresholds. The crucial difference though is that we explicitly leak the identity of the gate when it is affected by a fault. Such gates are then called \emph{stochastically compromised}.

Our solution to combat imperfections is then constructed in two broad steps by taking advantage of the composability of previous SDQC protocols (mostly~\cite{LMKO21verifying,KKLM22unifying,KKLM23asymmetric}). To this end, we start by recognizing that the robustness and security issues that we aim to address are linked to imperfections on the Verifier's side only. For the Prover's side, known protocols already provide security for arbitrary malicious adversaries ---~hence, against noisy ones~--- and robustness can always be dealt with by the Prover performing a standard fault-tolerant quantum computation. The immediate consequence is that we can focus on the problem of achieving robust and secure \emph{Remote State Preparation} (RSP) as it is only during the construction of this resource used by the SDQC protocol that the Verifier might manipulate quantum states. Once this is achieved, the second easy step is to recombine all the errors of the subprotocols with that of the main protocol. We apply this methodology to two subprotocols for constructing RSP resources, hence obtaining two protocols for SDQC with an imperfect Verifier.

Our first SDQC protocol deals with a simplified imperfection model where the operations on the Verifier's side are always perfect but stochastically leak the value of their classical (secret) control bits. This model is simpler but nonetheless corresponds to leakage in physical systems such as photonic Verifiers. The practical interest of our solution lies in the absence of physical overhead on the Verifier's side besides repeating the same sort of operations than if it was performing an insecure protocol. All of the overhead is taken care of by the possibly malicious Prover. This makes it an appealing setup for the near(er) term where very limited cheap Verifiers are envisioned. In addition, it also shows how to remove secret dependent stochastic noise which can be useful in other contexts. It does however impose a trade-off between the length of the computation and the security level that can be achieved as soon as the Verifier's actions are not perfect. The protocol for generating near-perfect states in the $\X-\Y$ plane form leaky ones is informally described in Protocol~\ref{proto:state_prep_informal}, assuming that a party called the Sender wants to send a state $\ket{+_theta} = \frac{1}{\sqrt{2}}(\ket{0} + e^{i\theta}\ket{1})$ to another called the Receiver.

\begin{protocolenv}[Plugged Rotated Remote State Preparation (Informal)]\label{proto:state_prep_informal}
\begin{algorithmic} [0]
    \STATE \begin{itemize}
    \item The Sender sends $N$ -- possibly leaky -- qubits $\ket{+_{\theta_j}}$ with the $\theta_j$ sampled randomly from $\{0, \pi/4, \ldots, 7\pi/4\}$.
    \item The Receiver applies a $\CNOT$ gate between each qubit $j \neq N$ (control) and qubit $N$ (target) and measures qubit $j$ in the computational basis, with measurement outcome $t_j$.
    \item At this stage, if both parties acted honestly, the Receiver is in possession of a qubit in the state $\ket{+_{\theta'}}$ with $\theta' = \theta_N + \sum_{j \in \{1, \ldots, N-1\}} (-1)^{t_j} \theta_j$. The Sender then sends the correction $\theta - \theta'$, which the Receiver uses as the angle for a $\Z$-axis rotation to recover the correct state $\ket{+_\theta}$.
    \end{itemize}
  \end{algorithmic}
\end{protocolenv}

Intuitively, the only information received during the protocol by the Receiver about the Sender's desired angle is the correction $\theta - \theta'$. There, each $\theta_i$ acts as a One-Time Pad for the angle $\theta$, meaning that the Receiver must know all the values $\theta_i$ to recover $\theta$. If $p_l$ is the probability of a single value $\theta_i$ leaking, then this happens with probability $p_l^N$, meaning that we have effectively exponentially suppressed the leak probability. We prove this in the composable Abstract Cryptography security framework, meaning that we can easily reuse this gadget in larger protocols while preserving its security guarantees.

We apply this state generation procedure to a SDQC protocol for classical inputs/outputs which requires $n$ repetitions of an MBQC computation over a graph of $|V|$ vertices. This protocol's security guarantee -- in particular the probability of accepting an incorrect computation is exponentially low in $n$ -- holds as long as all qubits sent by the Verifier are in the $\X-\Y$ plane with an angle hidden from the Prover. The leak probability of the overall protocol is $n|V|p_l^N$. 

Our second SDQC protocol deals with stochastically compromised gates. Because it needs to be robust, it must abide by the design principles of fault-tolerant quantum computation. We show that one can add stricter constraints to also obtain security. This is where the introduced modelling of the imperfection plays a crucial role as it allows to devise a way to wash away the information that could be stored in the syndrome before the logical qubits prepared by the Verifier are sent to the Prover. The challenge is of course to ensure that this is always the case in spite of the operations being possibly erroneous. 

More concretely, we again devise a scheme which the Verifier can use to prepare and send a qubit in the $\X-\Y$ plane to the Prover. These qubits are encoded in a concatenated Reed-Muller code and can be used by the Prover again to run the SDQC protocol mentioned earlier. The preparation phase consists of the creation of an encoded $\ket{+}$ state followed by a transversal rotation around the $\Z$ axis by an angle $\theta$. This step is tricky in this context because not only can the faulty $\Z$ rotation leak the value of $\theta$ but also because a fault in an earlier step might leak information as well. For example, an $\X$ error on a single qubit of the encoded $\ket{+}$ state becomes a $\Y$ error on the final state if $\theta = \pi/2$, but remains an $\X$ error if $\theta = 0$. This means that the adversary, which both controls what faults are applied and receives the final state can recover information about $\theta$ even if no direct leaks happened.

We therefore have to introduce a \emph{safe} transversal $\Z(\theta)$ simulation that can then be concatenated. This allows us to show that our construction satisfies two additional properties on top of the usual fault-tolerant requirements -- \emph{accuracy} and \emph{privacy} -- with doubly exponentially low failure probability with the level of concatenation. Accuracy simply captures that applying the secure version of gates should have the same effect as the original gate that the Verifier wanted to apply. Privacy captures the concern described above, meaning that the syndrome should not contain secret-dependent information. To satisfy the latter requirement, we split the application of each $\Z(\theta)$ gate into two gates $\Z(\alpha)\Z(\beta)$ such that $\alpha + \beta = \theta$. Again, the two values of $\alpha$ and $\beta$ for $\theta$ to be revealed. Together with a more involved error-correction step which tests, we prove that this ensures that the probability of leaking $\theta$ decreases for each level of concatenation by reproving a threshold theorem that can indeed be used to ensure these properties for arbitrary long but finite computations. 

Efficiency and scalability follow from the obtained theorem as for a fixed computation length, the protocol's error is negligible in the number of repetitions and the size of the single logical qubit registers, and for a fixed protocol error, an increase in the length of the computation induces only a polylogarithmic increase of the size of the single logical qubit registers.

\subsection{Prior work}
Secure delegation of quantum computation was first introduced in~\cite{C01secure,AS03blind} with a definition centred on blindness rather than verification. In 2004, the question of verification was posed by D.~Gottesman before being publicized in~\cite{A07scott}. Two approaches were then used to arrive at full verification: one based on quantum authentication schemes~\cite{ABE08interactive,ABEM17interactive}, and one relying on measurement based quantum computation~\cite{BFK09universal,FK17unconditionally}.

These initial works have set the stage for the following decade. Experiments were carried out using simplified settings~\cite{BKBF12demonstration,GRBM16demonstration,BFKW13experimental} while more protocols were introduced, extending the applicable setups as well as refining the techniques and improving on their performance. Most of these advances have been collected in the referenced review~\cite{GKK19verification}. In addition, other notions of security emerged such as flow-ambiguity~\cite{MDMF17flow-ambiguity}; quantum protocols served to improve the security of classical computing~\cite{DKK16quantum}.

Already within the earliest works~\cite{ABE08interactive,BFK09universal} the importance of noise robustness was recognized. It was however only through~\cite{ABEM17interactive} that the difficulty of the task was fully acknowledged and stated as an open question outside of a smaller circle of researchers working directly on the topic. The first work to provide a partial answer is~\cite{KD19nonadaptive} which brings verification and fault-tolerance together for specific non-adaptive architectures.

A different approach was initiated in~\cite{KLMO24verification} and ultimately led to the current work. There, the noise was restricted to be of constant strength per computation. Although not satisfactory from a practical perspective, it also managed to use composability to break the protocol into smaller blocks that could be studied and modified more easily. This led to a framework for producing modular SDQC protocols in~\cite{KKLM22unifying} giving flexibility to new ways to probe the honesty of the Prover, by taking an approach centred on error correction ---~hence benefiting from new theoretical tools. Finally, we also rely on a protocol that was introduced in~\cite{KKLM23asymmetric} for multi-party computation  with the specificity of only requiring the Verifier to prepare single qubit states in the equatorial plane of the Bloch sphere to yield verification.

Other works have been interested in robustness of delegated quantum computations, but without being fully satisfactory because of added assumptions that are difficult to ensure, or because they only partially solve the problem. In~\cite{GKW15robustness}, the proposed solution is in particular tied to the independence of the noise with respect to the secret parameters. In~\cite{GHK18simple}, verification is provided but blindness is not possible, while also requiring the errors to be uncorrelated between various rounds of the protocol.

The protocol introduced by~\cite{FH17verifiable} is in the Receive-and-Measure setting, in which the Prover prepares a large entangled state first and sends it qubit by qubit to the Verifier. The Prover prepares a state used for the Verifier's computation $N = 2k + 1$ times, of which $2k$ are used to test the Prover's honesty, while the remaining one is used for the Verifier's computation. Automatically, this implies that the security bound is at most an inverse polynomial of the number of runs. Indeed, the probability for the Prover to cheat in their protocol is $N^{-1/4}$. On the other hand, our protocol is secure with a negligible security error in the number of runs $N$ performed.

Furthermore, they state that they do not assume any model for the noise, and any noise applied by the Verifier can be seen as noise applied by the Prover on the resource state. They do specify that if the noise level is dependent on measurement bases chosen by the Verifier, it is possible to add more noise so that it becomes independent. However, it is not as simple. Indeed, the measurement bases chosen by the client are precisely their secret parameters for the tests and the computation. We will see that, in order to model correctly the noise which depend on these secret parameters, it is necessary to give to the Prover control over the type of noise applied, potentially leaking secret values altogether. This is precisely the pain-point that our construction is designed to mitigate.

Another protocol in the Receive-and-Measure setting is presented in~\cite{TFMI17fault}. There, the Verifier receives states and performs only physical $\X$ and $\Z$ measurements and classical post-processing on qubits encoded in a CSS code to prepare states in the $\X-\Y$ plane and $\Z$ eigenstates. We describe in Appendix~\ref{app:attack-tomoyuki} an attack on their state preparation protocol, in which we show that the Prover can influence the outcome of the final state by applying an $\X$ operation conditioned on the fact that the prepared state is in the logical computational basis. All proofs of verification assume that the deviation is independent of the secret parameters, meaning that this state preparation protocol cannot be combined with any existing SDQC protocol and remain secure. We refer to the appendix for more details. Finally, the same assumption of secret-independent measurement basis choice, as in~\cite{FH17verifiable}, is present and not addressed in this one.

\paragraph{Organisation of the paper.}
In \S~\ref{sec:preliminaries}, we present preliminaries on the Abstract Cryptography security framework and SDQC protocols in the MBQC model.
\S~\ref{sec:modeling} is devoted to motivating and constructing an appropriate model describing the imperfections of physical operations that we then use to describe noisy parties in our protocols.
\S~\ref{sec:plugging-rsp} describes and proves the security of two protocols for constructing Remote State Preparation (RSP) resources that can then be used to secure delegated computations in the presence of imperfect devices on the Verifier's side.
\S~\ref{sec:composing} uses these gadgets to derive two SDQC protocols and their security proofs.

\section{Preliminaries}
\label{sec:preliminaries}
\subsection{Notations}\label{sec:notation}

For any positive integer $n \in \NN^*$, $[n] := \{1, \ldots, n\}$ denotes the set of positive integers up to $n$. 
A function $\epsilon$ is \emph{negligible in $\eta$} if, for every polynomial $p$, for $\eta$ sufficiently large it holds that $\epsilon(\eta) < 1/p(\eta)$. 
The message $\Abort$ signals that a party aborted a protocol.We define the rotation operator around the $\Z$-axis of the Bloch sphere by an angle $\theta$ as $\Z(\theta) = \begin{pmatrix} 1 & 0 \\ 0 & e^{i\theta} \end{pmatrix}$ and $\ket{+_{\theta}} = \Z(\theta)\ket{+}$.
We denote by $\Theta$ the set of angles $\qty{\frac{k\pi}{4}, \ k = 0, \ldots, 7}$ .

\subsection{Abstract Cryptography Framework}
\label{subsec:AC}
The Abstract Cryptography (AC) security framework \cite{MR11abstract-cryptography,M12constructive-cryptography} used in this work follows the \emph{ideal/real simulation paradigm}. A protocol is considered secure if it is a good approximation of an ideal version called a \emph{resource}. Its main advantage is that any system that follows the structure defined by the framework is inherently composable, in the sense that if two protocols are secure separately, the framework guarantees at an abstract level that their sequential or parallel execution is also secure.  We refer the reader to \cite{DFPR14composable} for a more in-depth presentation, in particular regarding the framework's composability in the context of SDQC.

In this framework, the purpose of a secure protocol $\pi$ is, given a number of available resources $\mathcal{R}$, to construct a new resource -- written as $\pi \mathcal{R}$.  This new resource can be itself reused in a future protocol.  A resource~$\mathcal{R}$ is described as a sequence of CPTP maps with an internal state.  It has \emph{input and output interfaces} describing which party may exchange states with it.  It works by having each party send it a state (quantum or classical) at one of its input interfaces, applying the specified CPTP map after all input interfaces have been initialised and then outputting the resulting state at its output interfaces in a specified order. An interface is said to be \emph{filtered} if it is only accessible by a dishonest player. The actions of an honest player $i$ in a given protocol is also represented as a sequence of efficient CPTP maps $\pi_i$ -- called the \emph{converter} of party~$i$ -- acting on their internal and communication registers. We focus here on the two-party setting, in which case $\pi = (\pi_1, \pi_2)$.

In order to define the security of a protocol, we need to give a pseudo-metric on the space of resources.  We consider for that purpose a special type of converter called a \emph{distinguisher}, whose aim is to discriminate between two resources $\mathcal{R}_1$ and $\mathcal{R}_2$, each having the same number of input and output interfaces.  It prepares the input, interacts with one of the resources according to its own (possibly adaptive) strategy, and guesses which resource it interacted with by outputting a single bit.  Two resources are said to be indistinguishable if no distinguisher can guess correctly with good probability.

\begin{definition}[Statistical Indistinguishability of Resources]
  \label{def:ind-res}
  Let $\epsilon > 0$, and let $\mathcal{R}_1$ and $\mathcal{R}_2$ be two resources with same input and output interfaces.  The resources are \emph{$\epsilon$-statistically-indistinguishable} if, for all unbounded distinguishers $\mathcal{D}$, we have:
  \begin{equation}
    \label{eq:dist}
    \Bigl\lvert\Pr[b = 1 \mid b \leftarrow \mathcal{D}\mathcal{R}_1] - \Pr[b = 1 \mid b \leftarrow \mathcal{D}\mathcal{R}_2]\Bigr\rvert \leq \epsilon.
  \end{equation}
  We then write $\mathcal{R}_1 \underset{\epsilon}{\approx} \mathcal{R}_2$.
\end{definition}

The construction of a given resource $\mathcal{S}$ by the application of protocol $\pi$ to resource $\mathcal{R}$ can then be expressed as the indistinguishability between resources $\mathcal{S}$ and $\pi \mathcal{R}$.  More specifically, this captures the correctness of the protocol.  The security is captured by the fact that the resources remain indistinguishable if we allow some parties to deviate in the sense that they are no longer forced to use the converters defined in the protocol but can use any other CPTP maps instead.  This is done by removing the converters for those parties in Equation~\eqref{eq:dist} while keeping only $\pi_H = \prod_{i \in H} \pi_i$ where $H$ is the set of honest parties.  The security is formalised as follows in Definition \ref{def:ac-sec} in the case of two parties.

\begin{definition}[Construction of Resources]\label{def:ac-sec}
  Let $\epsilon > 0$. We say that a two-party protocol $\pi$ $\epsilon$-statistically-constructs resource $\mathcal{S}$ from resource $\mathcal{R}$ if:
\begin{enumerate}
\item It is correct: $\pi \mathcal{R} \underset{\epsilon}{\approx} \mathcal{S}$;
\item It is secure against malicious party $P_i$ for $i \in \{1, 2\}$: there exists a \emph{simulator} (converter) $\sigma_i$ such that $\pi_j\mathcal{R} \underset{\epsilon}{\approx} \mathcal{S} \sigma_i$, where $j \neq i$.
\end{enumerate}
\end{definition}

\subsection{Blind Delegation of Quantum Computations}
\label{subsec:prelims-ubqc}
The MBQC model of computation emerged from the gate teleportation principle. It was introduced in~\cite{RB01one} where it was shown that universal quantum computing can be implemented by first preparing a graphs state and later performing adaptive single qubit measurements on this state. This model has the same power as the usual gate-based model of quantum computations. The measurement calculus from~\cite{DKP07measurement-calculus} expresses the correspondence between the two models. We restrict ourselves to computations with classical inputs and outputs.

The Verifier's computation is defined by a \emph{measurement pattern} as follows.
\begin{definition}[Measurement Pattern]
  \label{def:pattern}
  A \emph{pattern} in the MBQC model is given by a graph $G = (V,E)$, input and output vertex sets $I$ and $O$, a flow $f$ which induces a partial ordering $\preceq$ of the vertices $V$, and a set of measurement angles $\{\phi(i)\}_{i\in V}$, where it is sufficient to have $\phi(i) \in \Theta$ for approximate universality~\cite{FK17unconditionally}.
\end{definition}

The flow is a function from non-output vertices $O^c$ to non-input vertices $I^c$ whose structure describes how measurement results influence future measurement bases. The flow-corrected measurement bases will be denoted $\phi'(i)$. The existence of a flow ensures that the computation is the same regardless of the measurement outcomes. Further details regarding the definition of the flow can be found in references~\cite{HEB04multiparty,DK06determinism}.

The computation is performed by creating the graph state associated to graph $G$ and then measuring each vertex $i$ according to the order $\preceq$ in the basis $\{\ket{+_{\phi'(i)}}, \ket{-_{\phi'(i)}}\}$. The output of the computation is given by the outcomes of measurements on the output vertices $O$, after corrections by the flow are applied.

If the Verifier prepares single qubits and uses quantum communication, it can delegate the computation of its choice to a quantum Prover blindly~\cite{BFK09universal}, meaning that the Prover does not learn anything about the computation besides the prepared graph $G$ and the order of measurements. To do so, the Verifier sends a state $\ket{+_{\theta(i)}}$ for each vertex $i$ in the graph, with $\theta(i)$ being chosen uniformly at random from the angle set $\Theta$. The Prover entangles these qubits according to the Verifier's desired graph. The Verifier then sends the angles $\delta(i) = \phi'(i) + \theta(i) + r(i)\pi$, for a random bit $r(i)$, in an order compatible with that of the flow of its computation. The Prover measures in the basis $\{\ket{+_{\delta(i)}}, \ket{-_{\delta(i)}}\}$ and returns the measurement outcome $b(i)$. The Verifier computes the correct measurement outcome $s(i) = b(i) \oplus r(i)$. Intuitively, the angle $\theta(i)$ perfectly hides $\phi'(i)$ and $r(i)$ perfectly hides the outcome $s(i)$. The correct computation is performed through the encryption since the angles $\theta(i)$ included in the state $\ket{+_{\theta(i)}}$ and the angle $\delta(i)$ cancel out.

We model the Verifier preparing and sending the state $\ket{+_{\theta(i)}}$ with the following Remote State Preparation (Resource \ref{res:rsp}). The Sender instructs it to prepare one state $\phi$ in a set of states $\Phi$ and the resource outputs the corresponding state at the Receiver's interface. Anticipating on the next sections, we consider this the ideal RSP since it never leaks more information about the identity of $\phi$ than that which is contained in the state of a single qubit quantum system whose state is $\phi$, and the output state is not impacted by any noise.

\begin{resourceenv}[Remote State Preparation]
\label{res:perf}
\label{res:rsp}
\item 
\begin{algorithmic}[0]

  \STATE \textbf{Sender's Input:} The classical description of a state $\phi \in \Phi$.

  \STATE \textbf{Computation by the Resource:} The Resource prepares the state $\ket{\phi}$ and sends it to the Receiver.
\end{algorithmic}
\end{resourceenv}

The AC security of UBQC can then be reduced to the AC security of this functionality with set of states $\Phi = \{\ket{+_\theta}\}_{\theta \in \Theta}$. Thanks to the composability of the AC framework, once we have constructed this RSP resource, it is possible to directly use this implementation to perform the UBQC protocol securely as well. The full UBQC protocol making use of this resource is presented in a formal way in Protocol~\ref{proto:ubqc}, Appendix~\ref{app:ubqc}.

\subsection{Dummyless SDQC Protocol for \texorpdfstring{$\BQP$}{BQP} Computations}
\label{subsec:dummyless-sdqc}
SDQC protocols based on the MBQC model used to require so-called \emph{dummy qubits} -- i.e. qubits in the state $\ket 0$ or $\ket 1$. Inserting such states in a graph is equivalent to removing from the graph the vertex where is was inserted along with all associated edges. This allows to isolate $\ket{+_\theta}$ qubits -- by removing all neighbouring vertices -- and use them as deterministic tests hidden from the Prover.

Recently, \cite{KKLM22unifying} introduced a framework for verifying quantum computations in the MBQC model that allows more general tests. Then \cite{KKLM23asymmetric} used this framework to to verify $\BQP$ computations without dummy qubits. This strategy is one of the key ingredients in our fault-tolerant SDQC protocol.

In the following, we summarise their construction of dummyless tests for arbitrary computation graphs $G$ and then give the associated SDQC protocol and security result.

\paragraph{Constructing Dummyless Tests for $\BQP$ Computations.}
The tests from~\cite{KKLM22unifying} are based on graph stabilisers.  Given a subset of vertices $W \subseteq V$, the test associated to $W$ requires the Verifier to prepare and send a $+1$ eigenstate of
\begin{equation*}
\bigotimes_{i \in W_{even}} \X \bigotimes_{j \in W_{odd}} \Y \bigotimes_{k \in N_G^{odd}(W)} \Z,
\end{equation*}
where $W_{even}$ (resp. $W_{odd}$) are the qubits of even (resp. odd) degree within $W$, and $k\in N_G^{odd}$ means $k$ is in the odd neighbourhood of $W$. This can be achieved by preparing all qubits in $W_{even}$ in a $+1$ eigenstate of $\X$, all qubits in $W_{odd}$ in a $+1$ eigenstate of $\Y$ and all qubits in $N_G^{odd}(W)$ in a $+1$ eigenstate of $\Z$.  A test accepts if the sum of the $\X$-measurement outcomes on qubits from set $W$ has parity $0$ -- it rejects otherwise.

Hence, whenever $N_G^{odd} = \emptyset$ the test can be prepared without using dummies. For instance, this is the case for $W = V$, but also for any set $W = V \setminus H$ where $H$ contains only even degree vertices. In fact, more generally, this holds true when $H$ is made of any number of even degree vertices together with chains of vertices starting and finishing on odd degree ones and having only even degrees in between. We denote $\mathfrak{W}$ the set of all such tests.

For classical input and output UBQC computations, it can be shown that all deviations be the Prover can be reduced to the application of a $\Z$ operation on a subset of vertices right before measuring in the $\ket{\pm_{\delta(i)}}$ basis.  A test detects one such deviation if the test rejects if the deviation is applied right before the measurement specified by the test.  A set of tests detects a given $\Z$ deviation if there is a non-negligible probability that a test sampled uniformly at random from the set detects the deviation.

In~\cite{KKLM22unifying}, they show that the set of tests $\mathfrak{W}$ can detect all $\Z$ deviations from the Prover apart from the deviation corresponding to applying a $\Z$ to all the odd degree qubits before measurement.  However, they also show that this undetected deviation has no effect on MBQC computations with classical outputs.  Hence all harmful deviations are detected and this set of tests can be used in an SDQC protocol as explained below.

\paragraph{The Dummyless SDQC Protocol.}
The principle of this protocol is to interleave UBQC executions of pure computation rounds and pure test rounds. The test in each test round is sampled at random from the dummyless tests constructed above. The final result is obtained by taking a majority vote on the obtained result conditioned to the ratio of failed test rounds being below a given threshold -- otherwise the computation is rejected.  This is presented formally below in Protocol~\ref{proto:dummyless-sdqc}.

\begin{protocolenv}[Dummyless SDQC Protocol for $\BQP$ Computations]
  \label{proto:dummyless-sdqc}
  \item 
  \begin{algorithmic}[0]  

    \STATE \textbf{Public Information:} 
    \begin{itemize}
    \item $G = (V, E, I, O)$, a graph with input and output vertices $I$ and $O$ respectively;
    \item $\mathfrak{W}$, the set of dummyless tests on graph $G$;
    \item $\preceq_G$, a partial order on the set $V$ of vertices;
    \item $N, d, w$, parameters representing the number of runs, the number of computation runs, and the number of tolerated failed tests.
    \end{itemize}
    
    \STATE \textbf{Verifier's Inputs:} A set of angles $\{\phi(i)\}_{i \in V}$ and a flow $f$ which induces an ordering compatible with $\preceq_G$.
    
    \STATE \textbf{Protocol:}
    \begin{enumerate}
    \item The Verifier samples uniformly at random a subset $C \subset [N]$ of size $d$ representing the runs which will be its desired computation, henceforth called computation runs.
    \item For $k \in [N]$, the Verifier and Prover perform the following:
      \begin{enumerate}
      \item If $k \in C$, the Verifier sets the computation for the run to its desired computation $(\{\phi(i)\}_{i \in V}, f)$. Otherwise, the Verifier samples uniformly at random a test $W$ from the set of dummyless tests $\mathfrak{W}$.
      \item The Verifier and Prover blindly execute the run using the UBQC Protocol~\ref{proto:ubqc}.
      \item If it is a test, the Verifier computes the XOR of measurement outcomes on the vertices from set $W$. The test accepts if it is $0$ and fails otherwise.
      \end{enumerate}
    \item At the end of all runs, let $x$ be the number of failed tests. If $x \geq w$, the Verifier rejects and outputs $\bot$.
    \item Otherwise, the Verifier accepts the computation. It performs a majority vote on the output results of the computation runs and sets the result as its output.
    \end{enumerate}  
  \end{algorithmic}
\end{protocolenv}

The security properties of SDQC protocols is captured by Resource~\ref{res:dqc}, which either produces the correct state or aborts depending on the choice of the Prover. The only tolerated leakage is denoted $l_{\cptp U}$, which depends on the Verifier's computation $\cptp U$.

\begin{resourceenv}[Secure Delegated Quantum Computation with Classical Inputs and Outputs]
  \label{res:dqc}
  \item
  \begin{algorithmic}[0]
  
	\STATE \textbf{Public Information:} The nature of the authorized leak $l_{\cptp U}$.
  	
    \STATE \textbf{Inputs:} 
    \begin{itemize}
    \item The Verifier inputs the classical description of a unitary $\cptp U$ over $m$ qubits.
    \item The Prover chooses whether or not to deviate. This interface is filtered by two control bits $(e, c)$.
    \end{itemize}
    
    \STATE \textbf{Computation by the Resource:}
    \begin{enumerate}
    \item If $e = 1$, the Resource sends the authorized leak $l_{\cptp U}$ to the Prover's interface and awaits further input from the Prover; if it receives $c = 1$, the Resource outputs $\Abort$ at the Verifier's output interface.
    \item If $c=0$, it outputs $O = \mathcal{M}_C \circ \cptp U \ket{0}^{\otimes m}$ at the Verifier's output interface, where $\mathcal{M}_C$ is a computational basis measurement.
    \end{enumerate}
  \end{algorithmic}
\end{resourceenv}

We can now give the security result for the dummyless SDQC protocol. Intuitively, so long as as we have a secure RSP implementation (Resource \ref{res:rsp}), the protocol realises the SDQC Resource up to an error which is negligible in the number of repetitions in the protocol. This theorem and the composability of the AC framework allow us to focus on the construction of the RSP Resource.

\begin{theorem}[Security of Protocol \ref{proto:dummyless-sdqc}, \cite{KKLM22unifying,KKLM23asymmetric}]
\label{thm:sec-dummyless}
Let $d$ be proportional to $N$, and let $c$ the fixed bounded error of the $\BQP$ class of computations.  Let $\epsilon$ be the (constant) lower bound on the probability that a test sampled uniformly from set $\mathfrak{W}$ rejects in the presence of a $\Z$ error.  Let $w$ be the maximum number of test rounds allowed to fail, chosen such that $w < \frac{2c-1}{2c-2}(N-d)(1 - \epsilon)$.  Let the leak be defined as $l_{\cptp U} = (G, \mathfrak{W}, \preceq_G)$.
  
Then Protocol~\ref{proto:dummyless-sdqc} $\eta(N)$-constructs the Secure Delegated Quantum Computation (Resource~\ref{res:dqc}) from the Remote State Preparation Resource~\ref{res:rsp} in the Abstract Cryptography framework, for $\eta(N)$ negligible in $N$.
\end{theorem}

\subsection{Fault-Tolerant Quantum Computation}
\label{sec:ftqc}
By their very nature, quantum computers manipulate delicate state superpositions and are thus inherently sensitive to noise. Without proper action,  noise quickly renders them useless.  Indeed, it is not enough for combating noise to perform error correction. The reason is that noise affects the error recovery procedure itself, thereby jeopardizing the ability to correct anything in the first place. In other words, reliable computations cannot be obtained from error correction alone.

To overcome this difficulty, fault-tolerance adds constraints and design rules so that reliable computation can be performed for circuits of arbitrary length, provided the noise level is not too high. While there is a plethora of fault-tolerant schemes that can be used in various situations (e.g.~\cite{BM09quantum,CB18flag,NB18measurement,HH21dynamically}), with specific advantages, we will only present here a very high-level description of this topic following~\cite{AGP06quantum} which builds upon the original ideas of \cite{AB97fault}. The reason is two-fold:
\begin{itemize}
\item Our construction will rely heavily on its proof of the threshold theorem;
\item \cite{AGP06quantum} is an accessible, yet thorough and rigorous proof of the threshold theorem making it a resource of choice for additional details.
\end{itemize}

\paragraph{Modelling Imperfections.}
The first ingredient of the fault-tolerance theorem is the modelling of devices whose imperfections need to be handled by the procedure. The crucial realization put forth in~\cite{AGP06quantum} is that one can reach fault-tolerance without relying on the tiny details of the error model that affect the device. Rather, it relies on the stochastic and independent nature of the \emph{location} of faults, while the error applied on the affected qubit can be chosen in an arbitrary malicious way --- e.g.~it can be correlated with errors happening at other locations or be the result of interactions with a private quantum register held by an adversary, and thus need not be Markovian. The term \emph{location} in a quantum circuit refers to a wire in a preparation, a measurement, or a gate at the physical level. For example, the preparation of a single qubit counts as one location, but a two-qubit gate counts as having two locations, one for each qubit.

\paragraph{Requirements on Fault-tolerant Gadgets.}
The second ingredient consists of defining a set of requirements that, if met, guarantee the existence of a threshold for the probability of faults, under which arbitrarily long computations can be performed with a manageable overhead. These requirements stem from the choice of protecting quantum information using recursive single qubit block encoding into distance-3 error correcting codes.\footnote{The analysis in the rest of the paper would work for larger distances for the base code, but the conditions P0-P4 below would need to be adapted~\cite{AGP06quantum}.}

More precisely, it relies on the concept of recursive simulations in which an ideal gate $\mathsf{Ga}$ is simulated at \emph{level-0} by implementing it using the corresponding imperfect physical gate denoted $0-\mathsf{Ga}$. Otherwise said, a level-$0$ simulation of a whole circuit is simply the level-$0$ simulation of each of its individual gates. A first level of protection against errors can be obtained by encoding each qubit of the level-$0$ circuit into an error correcting code and performing an error correction step after each prescribed $0-\mathsf{Ga}$ of the level-$0$ circuit. More precisely, each $0-\mathsf{Ga}$ is replaced by a $1-\mathsf{Ga}$ followed by a $1-\mathsf{EC}$, together constituting a $1-\mathsf{Rec}$. Here, $1-\mathsf{Ga}$ stands for the gate implementation at the logical level, and $1-\mathsf{EC}$ is the error correction procedure applied to the encoded qubit. A level-$1$ simulation of a quantum circuit is obtained when all $0-\mathsf{Ga}$ have been replaced by their corresponding $1-\mathsf{Rec}$.  From there, a level-$k$ simulation is defined recursively and is obtained from a level-$(k-1)$ simulation by replacing all $0-\mathsf{Ga}$ by their corresponding $1-\mathsf{Rec}$'s. Figure~\ref{fig:bad_logical_Z} presents level-$0$ to level-$2$ simulations for a $Z(\theta)$ gate.

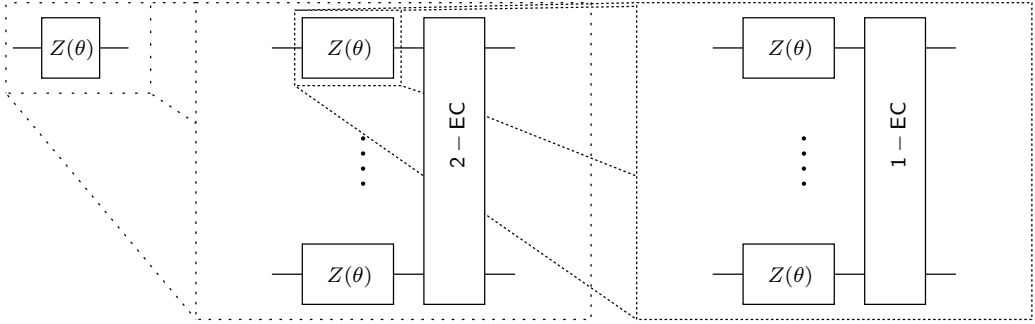
\begin{figure}
  \centering
  \begin{picture}(0,0)\includegraphics{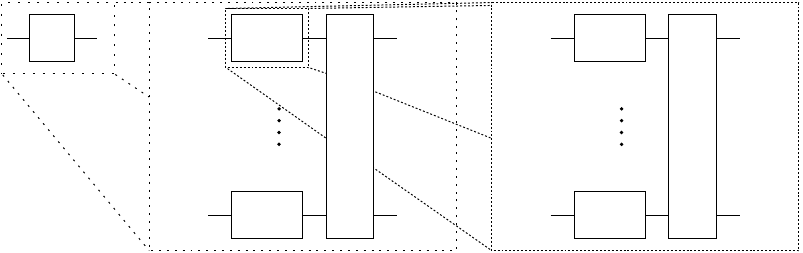}\end{picture}\setlength{\unitlength}{4144sp}\begin{picture}(6099,1914)(34,-973)
\put(2761,-61){\rotatebox{90}{\makebox(0,0)[lb]{\smash{\fontsize{8}{9.6}\usefont{T1}{ptm}{m}{n}{\color[rgb]{0,0,0}$\mathsf{2-EC}$}}}}}
\put(1946,614){\makebox(0,0)[lb]{\smash{\fontsize{8}{9.6}\usefont{T1}{ptm}{m}{n}{\color[rgb]{0,0,0}$Z(\theta)$}}}}
\put(1946,-736){\makebox(0,0)[lb]{\smash{\fontsize{8}{9.6}\usefont{T1}{ptm}{m}{n}{\color[rgb]{0,0,0}$Z(\theta)$}}}}
\put(301,614){\makebox(0,0)[lb]{\smash{\fontsize{8}{9.6}\usefont{T1}{ptm}{m}{n}{\color[rgb]{0,0,0}$Z(\theta)$}}}}
\put(5371,-61){\rotatebox{90}{\makebox(0,0)[lb]{\smash{\fontsize{8}{9.6}\usefont{T1}{ptm}{m}{n}{\color[rgb]{0,0,0}$\mathsf{1-EC}$}}}}}
\put(4556,614){\makebox(0,0)[lb]{\smash{\fontsize{8}{9.6}\usefont{T1}{ptm}{m}{n}{\color[rgb]{0,0,0}$Z(\theta)$}}}}
\put(4556,-736){\makebox(0,0)[lb]{\smash{\fontsize{8}{9.6}\usefont{T1}{ptm}{m}{n}{\color[rgb]{0,0,0}$Z(\theta)$}}}}
\end{picture}   \caption{Fault-tolerant Level-2 Simulation of a $\Z(\theta)$ Gate. Left: level-0 simulation corresponding to the logical level; middle: level-1 simulation; right: level-2 simulation depicting the implementation of the level-1 $\Z(\theta)$ for the first qubit.}
  \label{fig:bad_logical_Z}
\end{figure}

One can then express the prerequisite for obtaining the threshold theorem as: 
\begin{enumerate}
\item[(P0)] If a level-$1$ error correction contains no fault, it takes any input to an output in the code space.
\item[(P1)] If a level-$1$ error correction contains no fault, it takes an input with at most one error to an output with no errors.
\item[(P2)] If a level-$1$ error correction contains at most one fault, it takes an input with no errors to an output with at most one error.
\item[(P3)] If a level-$1$ gate contains no fault, it takes an input with at most one error to an output with at most one error in each output block.
\item[(P4)] If a level-$1$ gate contains at most one fault, it takes an input with no errors to an output with at most one error in each output block.
\end{enumerate}

\paragraph{Fault-tolerance Threshold Theorem.}
The third ingredient of fault-tolerance is the threshold theorem itself: 
\begin{theorem}[Restated from~\cite{AGP06quantum}]\label{thm:ft}
  When properties P0-P4 above are satisfied, there exists some $\varepsilon_0$  such that for a given $\varepsilon_{{\mathrm{corr}}}  > 0 $, any ideal circuit $C_0$ with $L$ fault locations and depth $D$ can be simulated with a failure probability at most $\delta$  by a circuit $C$  with at most $L^*$ fault locations and depth at most $D^{*}$,  where:
  \begin{equation}
    L^* = O(L(\log L)^{\log_2 l}), \quad  D^* = O(D(\log L)^{\log_2 d}).
  \end{equation}
Above $l$ is the maximum size number of fault locations in a $1-\mathsf{Rec}$ and $d$ is the maximum depth of a $1-\mathsf{Rec}$.
\end{theorem}
The  circuit $C$ that satisfies the above theorem is obtained by using $k$ levels of recursive encoding per qubit, with $k$ such  that $2^k \geq \frac{\log(2\varepsilon_0 L / \delta)}{\log(\varepsilon_0/\varepsilon_{\mathrm{corr}})}$.

As we will abundantly rely on the structure of its proof, it will be presented in \S~\ref{sec:rsp-ac} where it will be of direct use to the reader.

\paragraph{Concrete Implementations.}
And finally, the fourth step is the construction of fault-tolerant level-$1$ gadgets that abide properties P0-P4 above. Indeed, we will not delve in their construction here as this is not directly needed for our purposes. We will instead rely on their existence and proceed in two steps to produce a secure version of these gadgets. For all but the single qubit $\Z(\theta)$ rotations, we will add a trailing $1-\mathsf{EC}$s for each $1-\mathsf{Rec}$s. For the level-$1$ $\Z(\theta)$ it will employ a \emph{split compilation} of the logical gate (see \S~\ref{sec:rsp-ac}).

\section{Modeling Imperfect Operations}
\label{sec:modeling}
Imperfections on the Prover's devices are already taken into consideration in the security of SDQC protocols: since they are secure for arbitrary malicious Provers, they must also be secure against noisy Provers. On the other hand, imperfect Verifier devices may induce new opportunities for attacks which are not accounted for in existing protocols. In this section, we model imperfections for the devices used to perform the SDQC protocols and analyze their impact on the security of the SDQC protocol when they affect devices used by the Verifier.

We start by examining imperfections that yield a direct failure of security proofs as a result of partial or full leak of the value of secret key bits that are used in Protocol~\ref{proto:dummyless-sdqc}. These situations may arise due to the physical setup where the states prepared by the Verifier to delegate its computation are not perfect qubit states, but rather weak coherent pulses that are susceptible to photon splitting attacks. It can also be due to side-channel attacks which have a certain probability of successfully leaking the secret bits. We will model this with the \emph{Stochastically Leaky RSP} (Resource~\ref{res:stoca}), which captures all cases in which similar leaks can occur.

Then, we will examine the effect of imperfections that can be tamed using fault-tolerant techniques. This is because an SDQC protocol must be fault-tolerant to be correct in the presence of imperfect Verifier and Prover operations. Unfortunately, our analysis will reveal that standard fault-tolerant techniques applied in the context of these protocols intertwine correctness and security, making it impossible to separate faults from leaks. This fact was already noted in~\cite{ABEM17interactive} and left as an open difficulty on the road to verification of fault-tolerant delegated quantum computations. We therefore introduce \emph{Stochastically Compromised Preparations, Gates, and Measurements} which describe arbitrary combinations of stochastic leaks and faults. Constructing an SDQC protocol using this resource will guarantee that the same protocol will also both (i) be secure if the resource only leaks but is correct, and (ii) yield the correct result in the standard model of fault-tolerance in which there are no secret parameters, thereby achieving the goal of this paper.

\subsection{Effects of Leaks on Security}
\label{subsec:leak_model}

\paragraph{Leaks Occur in Physically Motivated Setups.}
To take a concrete example, consider the RSP that produces $\ket{+_\theta}$ states for $\theta \in \Theta$. A direct way to construct such resource is to consider that (i) the Sender has access to a single photon source that prepares $\ket{+}$ states, that (ii) it can perform $Z(\theta)$ rotations on the produced states, and that (iii) it can send them through an optical fiber to the Receiver.

However, even using very good heralded single photon sources, there is a non-negligible\footnote{Following the definition of \emph{negligible} given in \S~\ref{sec:notation}.} probability that more than one photon is emitted in the transmitted pulse.

This can also be the consequence of resorting to a cheaper setup that uses phase randomized weak coherent pulses instead of good single photon sources. In that case, strongly attenuated lasers generate a correct $\ket +$ state except that with some probability $p_{l}$ it prepares more than one photon in the pulse --- i.e. the state becomes $\ket{+}^{\otimes n}$ for some $n\geq 2$.

Whenever these implementations of RSP are used within Protocol~\ref{proto:dummyless-sdqc} to send $\ket{+_{\theta}}$, it is susceptible to the photon splitting attack~\cite{HIGM95quantum}. The secret values of $\theta$ that need to remain concealed can be partly recovered by the Prover. For instance, suppose the Verifier sends $\ket{+_{\theta}}^{\otimes 2}$ to the Prover instead of $\ket{+_{\theta}}$. An arbitrarily powerful Prover can detect that the sent pulse has two photons and can measure one of them in the $\ket{\pm_{\phi}}$ basis for some random $\phi \in \Theta$. If the outcome is $\ket{+_{\phi}}$, it knows that the value of $\theta$ cannot be $\phi + \pi$, while if the outcome is $\ket{-_{\phi}}$, then $\theta \neq \phi$. It can therefore learn information on $\theta$ without ever being detected by Protocol~\ref{proto:dummyless-sdqc}. Similar physically motivated imperfections can also yield side-channel attacks~\cite{GQL24hidden} without necessarily resorting to photon splitting.

In addition, there are other cases in which leakage needs to be taken into account. If the SDQC protocol is used to gauge the imperfections of a device, both the Verifier's and the Prover's operation may be implemented on the same device. In this case, there may be cross-talk between the parts of the device used to implement the two parties, leading to secret leakage. This must then be quantified and accounted for in the analysis~\cite{GLMM24chip,FBME24certified}. This is particularly true if some components are supposed to be in a trusted environment or enclave~\cite{MKAC22qenclave,KLMO24verification}

As an intermediate conclusion, such imperfections that leak (part of) the value of secret key bits are rather common and expected when using real world devices.

\paragraph{Impact on Security.}
The issue with leaking even partial information about the secret to the Prover is that it becomes impossible to twirl the Prover's attack in the security analysis of the SDQC protocol, upon which all such proof rely. Indeed, averaging over the values of the secret bits performs a twirl of the Prover's deviation~\cite{DCEL09exact,K16efficient}, reducing coherent attacks to convex combination of Paulis. This can be done as long as \emph{all} the key bits are \emph{perfectly secret}.

Technically, in the situation above, the Prover has the ability to condition its attack on the information gathered through the leaks. In particular, removing the full blindness guarantee thanks to the leaked information may allow a malicious Prover to adapt its behaviour to the testing strategy used to probe its the honesty and circumvent it.

\paragraph{Modelling Leaks.}
If we are to successfully prove security with imperfect devices, the secret parameters cannot leak all the time, otherwise there is simply no secret at all. It is thus necessary to find a model for the imperfections so that it captures the behaviour of physically realistic devices\footnote{Our model does not need to exactly correspond to physical imperfections as long as the real devices can be used to simulate the modelled imperfections. In such case, this means that if security can be obtained for the modelled imperfections, it will immediately imply security when it is replaced by real devices.}, while not leaking all the secret parameters at all times. Taking inspiration from the examples of leaks described above, it is rather natural to restrict the leaks to be stochastic. This then yields to the following resource definition.

\begin{resourceenv}[Stochastically Leaky Remote State Preparation]
\label{res:stoca}
\item
\begin{algorithmic}[0]

  \STATE \textbf{Public Information:} Leak probability $p_l$, $\Phi$ the set of states that can be prepared

  \STATE \textbf{Sender's Input:} The classical description of a state $\phi \in \Phi$.
  
  \STATE \textbf{Receiver's Input:} A bit $c \in \{0, 1\}$ indicating whether it wants to cheat, set to $0$ in the honest case.

  \STATE \textbf{Computation by the Resource:} The Resource prepares the state $\ket{\phi}$ and sends it to the Receiver. If $c = 1$, it flips a biased coin whose probability of outputting $1$ is $p_l$. If successful, it outputs as well the classical description $\phi$ at another interface belonging to the Receiver.
\end{algorithmic}
\end{resourceenv}

The main advantage of this definition is its simplicity: it either leaks the full classical description or nothing. This captures the worst case scenario in terms of possible leakage, therefore any protocol which is secure when built upon this resource will also be secure if only part of the state leaks. However, note that the state produced by the resource is always correct, meaning that it assumes that any error is applied only after the state has been produced. This does not capture all the situations in which the qubits suffer from errors during their preparation. We must therefore incorporate such possibiity into a more refined model.

\subsection{Effects of Faults on Security}
\paragraph{Faults Must Be Accounted for to Provide Robustness.}
The ``standard model'' for imperfections affecting quantum computing devices is defined in the context of fault-tolerant computations. As introduced in \S~\ref{sec:ftqc}, the crucial point to combat noise effectively and efficiently in quantum computers is to impose requirements on faults --- i.e. the locations of errors --- rather than on errors themselves. More precisely, the model used in fault-tolerance specifies that faults are stochastic, while errors conditioned on all the faults are arbitrarily malicious. This, for instance, covers far more adverse situations than using a fixed error model such as the depolarizing noise: it allows a malicious environment to create correlated errors at various locations once these are determined.

As we are proving the correctness and security of our protocols in the AC framework, it is natural to assume that the distinguishedthat is in charge of telleing apart the real and ideal scenarii in the AC framework has full control over the environment. In other words, the adversarial Prover is the one fixing the errors that occur during the whole execution, even on the machines of the Verifier.\footnote{An alternative would be to model this as two different Adversaries which may or may not collude. We choose the stronger version so that our protocol is secure in the worst case.}

\paragraph{Fault-Tolerant Gadgets Are Generally Insecure.}
We now describe an attack which stems from a naive implementation of fault-tolerant gates in the context of SDQC. This example will demonstrate the features required for our model of faulty operations presented below.

We consider again an RSP producing $\ket{+_{\theta}}$ states for $\theta \in \Theta$. Because we wish to protect the full computation from faults, this quantum state must be encoded in a quantum error-correcting code. The RSP on the Verifier's side then consists of preparing the logically encoded $\ket +$ state and performing the logical $\Z(\theta)$ rotation. Both operations can be performed fault-tolerantly. For instance, using the concatenated quantum Reed-Muller code $[[15,1,3]]$ the $\Z(\theta)$ rotation can even be implemented transversally using single qubit $\Z(\theta)$ rotations. We assume this setup is known to the Prover: the encoding and the method used for the RSP are part of the protocol's description, which the Verifier follows honestly. 

Recall that after the logical $\ket +$ preparation, physical qubits can be faulty, that the positions of faults are known to the adversary and that it may choose the errors on faulty locations. As the concatenation level grows there is a constant fraction of the physical qubits that are faulty with overwhelming probability. The adversarial Prover then acts as follows: (i) until the logical $\ket{+}$ state is prepared, it applies the identity $\Id$ as the error for any fault; (ii) at the location in the circuit right before applying the $\Z(\theta)$ gate, it decides to always apply a physical $\X$ as the malicious error on the faulty qubits. It then decides to again apply the identity at any fault during the $\Z(\theta)$ and after. After this process, the Verifier sends the prepared state to the Prover.

The problem can now be unveiled: at this stage, with overwhelming probability the information about $\theta$ can be read from the syndrome measurements of the code. This is because the $\X$ gates have distinct commutation relations with $\Z(\theta)$ depending on the value of $\theta$. For instance, if $\theta = \pi/2$ and $\Z(\theta) = \cptp S$ is the phase gate, $\X$ is transformed into $\Y$. Because the chosen code is a CSS code, it will trigger $\X$ and $\Z$ syndromes. On the other hand, if $\theta = 0$ and $\Z(\theta) = \Id$, only $\Z$ stabilizers will be triggered. This information will indeed be available to the Prover since, even after a perfect round of error correction is applied, some of the concatenated blocks will have uncorrectable errors. The Prover will get access to these blocks still containing the secret information once it receives the logical state from the Verifier. Hence, fault-tolerant gadgets are generally insecure.

Note that this exploit is a result of the very general behaviour of classically controlled gates: a fixed unitary applied before the gate that is commuted through the gate is mapped onto a unitary that \emph{generally depends} on the classical control. In other words, a \emph{secret-independent} noise is transformed by a secret dependent gate into a \emph{secret-dependent} noise, thus justifying the need for a unified approach of leaks and faults.

\paragraph{Compromised Operations.}
The discussion above outlines the necessity to model both leaks and faults on physical qubit operations within a single resource. In addition, it must be compatible with the standard model for fault-tolerance as it will also be used to model the imperfect operations on the Prover's side. We model physical preparations of quantum states, gates and measurements in the same resource for simplicity of presentation and coherence.

\begin{resourceenv}{Stochastically Compromised Preparation, Gate, or Measurement}
    \label{res:scpg}
    \label{res:compromised-pgm}
    \item

  \begin{algorithmic}[0]
  	
  	\STATE \textbf{Public Information:} 
	\begin{itemize}
		\item The compromise probability $p_c$;
		\item The size of the input and output registers $n$;
		\item The set of parameters that can be used to classically control the operation $\Lambda$, reduced to the name of the preparation, gate or measurement if not parametrized;
		\item The set of CPTP maps $\{\cptp U(\lambda)\}_{\lambda \in \Lambda}$.
	\end{itemize}

    \STATE \textbf{User's Input:}
	\begin{itemize}
	\item A value $\lambda \in \Lambda$ that parametrizes the preparation, gate or measurement;
	\item The quantum registers that are acted upon by the Resource.
	\end{itemize}	    

    \STATE \textbf{Eavesdropper's Input:} Bits $c_i \in \{0, 1\}$ indicating whether it wants to cheat, set to $0$ in the honest case and to $1$ when it wants to compromise the $i$\textsuperscript{th} location of the preparation, gate or measurement.

    \STATE \textbf{Computation by the Resource:}
    \begin{itemize}
    \item If the input register is non-empty, it inserts the state $\rho$ of the provided register into its internal register. Else it initialises its input register in the $\ket{0}^{\otimes n}$ state.
    \item It applies $\cptp U(\lambda)$ to $\rho$.
    \item For each $i$ s.t. $c_{i} =1$, with probability $p_c$:
        \begin{itemize}
          \item It sends $\lambda$ and the register at location $i$ to the Eavesdropper;
          \item It receives from the Eavesdropper a quantum system and inserts into its internal register at position $i$.
        \end{itemize}
    \item It outputs the state in its internal register at the User's output interface.
    \end{itemize}
  \end{algorithmic}
\end{resourceenv}

\paragraph{Applicability and Limitations.}
Both the leaky-but-correct case and standard fault-tolerant computations can be recovered from the imperfections modelled by Resource~\ref{res:compromised-pgm}.
The Stochastically Leaky RSP corresponds to the Stochastic Compromised Preparation with a non-trivial parametrization set $\Lambda$ where no error is applied by the Eavesdropper.
The standard model of fault-tolerance can be obtained by assuming that the Eavesdropper always cheats. In this case there is no parametrization and the authorized leak $\Lambda$ reveals the full computation. 

However, this Resource does not encompass all possible imperfections. In particular, we impose a binary behavior for the leak: it either leaks the full description of the gate or nothing. We will see in the next section that bounding away from zero the number of instances which do not leak, and therefore have perfect security, is a crucial requirement for the security of RSP. We leave as an open question how to deal with a Resource which always leaks a small amount of information about the chosen parameters $\lambda$.

\section{Exponentially Suppressing Imperfections in Remote State Preparation}
\label{sec:plugging-rsp}
\label{sec:lt}

In this section, we construct ideal RSP resources using the imperfect RSP resources defined in the previous section. The purpose of such resources is to produce a state on the Receiver's interface, chosen by the Sender from a set of states $\Phi$. At the most basic level this represents a device in the hands of the Sender that is able to produce various quantum states which are then sent via a quantum channel to the Receiver. When constructed in the AC framework, the composability of these resources together with a proper choice of $\Phi$ allows to replace their concrete implementations with their ideal version in order to implement UBQC (Protocol~\ref{proto:ubqc}).

In \S~\ref{sec:def-rsp}, the proposed protocol will require that the Sender performs single physical qubit operations modelled by the Stochastically Leaky RSP (Resource~\ref{res:stoca}). This simplicity comes at the price of robustness: while the security of the construction approaches perfection exponentially fast with the number of physical resources used, it degrades equally fast in the presence of noise. As a result, when it is used inside a verification protocol, it will force the Verifier --- playing the role of the Sender --- to trade-off between correctness and security (see \S~\ref{sec:issue-correctness}).

On the contrary, in \S~\ref{sec:rsp-ac} the protocol used to construct the RSP is robust to noise and allows for autonomous FTQC by the Prover. Yet, this comes at the expense of requiring the Verifier --- playing the role of the Sender --- to perform single logical qubit manipulations, modelled by Stochastically Compromised Preparations, Gates and Measurements (Resource~\ref{res:compromised-pgm}). As a result, it bypasses the trade-off encountered in the first proposal by resorting to a more complex protocol.

Both protocols will explicitly depend on a parameter that counts the number of physical resources they use. We show that they construct the RSP resource with a distinguishing advantage that is negligible in this parameter, thereby being efficient.

\subsection{Suppressing Leaks via Prover-side State Merging}
\label{sec:def-rsp}
\label{sec:rotated-rsp}

\paragraph{Protocol for Mitigating Leaks in RSP.}
Our goal in this section is to construct the Ideal RSP Resource~\ref{res:perf} using several calls to the leaky one (Resource~\ref{res:stoca}). More precisely, the Plugged Rotated Remote State Preparation uses $N$ calls to Resource~\ref{res:stoca} to provide a single $\ket{+_\theta}$ state on the Receiver's side with the guarantee that no leak about $\theta$ happens unless all $N$ calls are leak their secret parameter. To that effect, the Sender uses the Stochastically Leaky Rotated State Preparation Resource to send rotated $\ket{+_\theta}$ states to the Receiver. Upon reception, the Receiver applies a series of $\CNOT$'s between these qubits in a star-shaped pattern followed by a computational basis measurement of all but one qubit. Finally, the Receiver applies a correction on the unmeasured qubit. The formal version of the Plugged Rotated Remote State Preparation is given in Protocol \ref{proto:state_prep}, and it is pictured in Figure \ref{fig:st-prep} for $N = 8$ qubits.

\begin{protocolenv}[Plugged Rotated Remote State Preparation]
  \label{proto:state_prep}
\item
\begin{algorithmic} [0]

    \STATE \textbf{Sender's Input:} An angle $\theta \in \Theta$

    \STATE \textbf{Protocol:}
    \begin{itemize}
    \item The Sender samples uniformly at random $N$ angles from $\Theta$, denoted $(\theta_1, \ldots, \theta_N)$.
    \item The Sender and Receiver invoke $N$ times a Stochastically Leaky Remote State Preparation (Resource \ref{res:stoca}) for state set $\qty{\ket{+_\theta}}_{\theta \in \Theta}$ with leakage probability $p_l$. The Sender inputs $(\theta_1, \ldots, \theta_N)$ and the Receiver obtains $\bigotimes_{j \in [N]} \ket{+_{\theta_j}}$ along with leaks $(\theta_{i_1}, \ldots, \theta_{i_k})$.
    \item For each $j \neq N$, the Receiver applies $\CNOT_{N,j}$ between the qubits $N$ and $j$, with the first being the control and the second the target. It measures the target qubit $j$ in the computational basis with measurement outcome $t_j$.
    \item The Receiver sends the vector $t$ containing all the measurement outcomes to the Sender.
    \item The Sender computes $\theta' = \theta_N + \sum_{j \in [N-1]} (-1)^{t_j} \theta_j$ and sends a correction $(b, (-1)^b\theta - \theta')$ to the Receiver, who applies $\X^b\Z((-1)^b\theta - \theta')$ to the unmeasured qubit, keeping it as output.
    \end{itemize}
  \end{algorithmic}
\end{protocolenv}

\begin{figure}[ht]
  \centering
  \subfloat[The Receiver obtains the qubits and applies $\CNOT$ gates (the central qubit $N$ is the control, the rest are targets).]{
    \label{fig:st-prep-1}
    \includegraphics[width=0.4\textwidth]{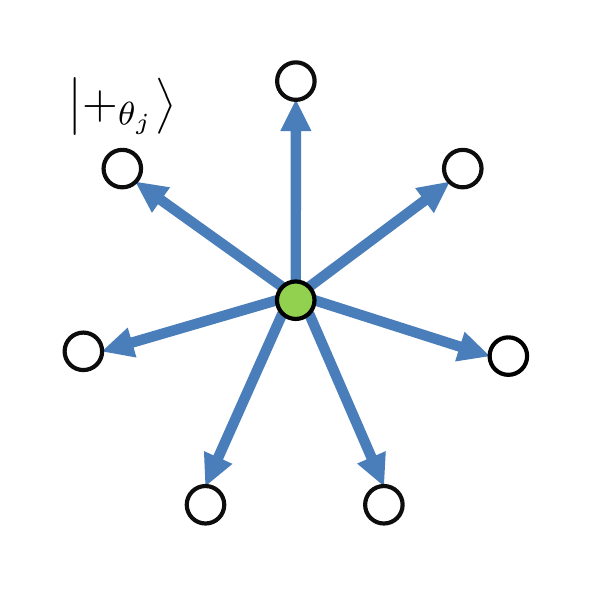}
  }
  \qquad
  \subfloat[The Receiver measures all qubits but the central one in the computational basis and gets outcomes $t_j \in \bin$.]{
    \label{fig:st-prep-2}
    \includegraphics[width=0.4\textwidth]{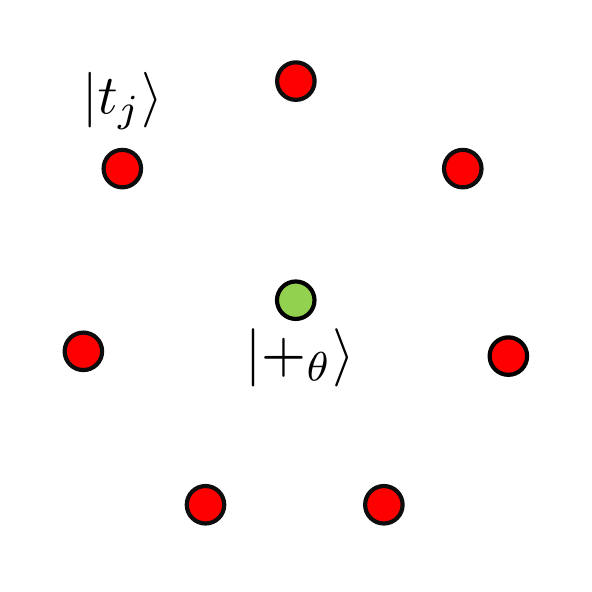}
  }
  \caption{Plugged Rotated Remote State Preparation for eight qubits. All qubits are in the state $\ket{+_{\theta_j}}$, supplied by a Stochastically Leaky Remote State Preparation Resource. The rotation angle $\theta$ on the central qubit after the protocol is completed corresponds to the one chosen by the Sender as its input.}
  \label{fig:st-prep}
\end{figure}

Notice that the Receiver needs to apply an additional correction $\X^b$ which at first glance seems unnecessary.\footnote{Indeed, if $b = 1$ the only effect of this $\X$ is to cancel the minus sign in front of $\theta$ in the other part of the correction.} This is required however because the simulator in the security proof will use teleportation to generate the correct state at the Receiver's interface. Correcting the teleportation may require the simulator to instruct the Receiver to apply a bit-flip, meaning that this instruction must also appear in the honest execution of the protocol so that the two may be indistinguishable.

\paragraph{Correctness and Security of the Protocol.}
We start by proving a lemma that will be useful when examining the security of our protocols. In essence, the Stochastically Leaky RSP resource is an imperfect realisation of the Ideal RSP functionality. Hence, any secure realisation of the former is also a secure realisation of the latter, where security decreases at most by a factor linear in the leak probability.

\begin{lemma}
\label{lem:stoc-to-perf}
Let $\mathcal{SL}_{p_l}$ and $\mathcal{I}$ be respectively the Stochastically Leaky Remote State Preparation with leak probability $p_l$ and the Ideal Remote State Preparation Resource, for the same state set $\Phi$. Then, with $\#\Phi$ denoting the cardinality of $\Phi$:
\begin{equation}
\mathcal{SL}_{p_{\mathit{leak}}} \underset{\epsilon}{\approx} \mathcal{I}, \mbox{ with } \epsilon = \left(1- \frac{1}{\#\Phi}\right)p_l.
\end{equation}
\end{lemma}

\begin{proof}
  The correctness is perfect as both resources produce the same input and outputs in the honest case. Regarding security, the Simulator against the malicious Receiver who has chosen $c = 1$ performs a call to the Ideal Remote State Preparation Resource \ref{res:perf} and receives a quantum state $\ket{\phi}$ for $\phi \in \Phi$. It tosses a coin with probability $p_l$ of obtaining $1$. If the result is $1$, it samples uniformly at random $\tilde{\phi} \in \Phi$ and sends it to the Server. In the real world and in the ideal world, the state $\ket{\phi}$ is obtained by the Receiver. The only difference occurs when the Simulator fails to provide the correct state description to the Receiver. The transcripts are otherwise perfectly identical. This happens with probability $(1 - 1/\#\Phi)p_l$ since it must both leak test and fail to guess the correct state from set $\Phi$, meaning that $\tilde{\phi} \neq \phi$. In such case, and only in such case, a Distinguisher can then tell apart the real and ideal world implementations as it knows the value of $\phi$ for having provided it at the Sender's interface. Therefore $(1 - 1/\#\Phi)p_l$ is also the distinguishing probability between the two implementations which concludes the proof.
\end{proof}

We can now state the main result of this section, namely that Protocol~\ref{proto:state_prep} boosts exponentially the probability of not leaking the state of the unmeasured rotated qubit.

\begin{theorem}[Exponentially-Boosted Rotated State Preparation]
  \label{thm:boost}
  Given the set of states $\Phi = \qty{\ket{+_\theta}}_{\theta \in \Theta}$, Protocol \ref{proto:state_prep} perfectly constructs in the AC Framework a Stochastically Leaky Remote State Preparation Resource $\mathcal{SL}_{p_l^N}$ for the set $\Phi$ using $N$ Stochastically Leaky Remote State Preparation Resources $\mathcal{SL}_{p_l}$ for the set $\Phi$.
\end{theorem}

The proof of the protocol is given by the next two lemmata that provide respectively the correctness and security of Protocol \ref{proto:state_prep}.

\begin{lemma}[Correctness of Plugged Rotated Remote State Preparation]\label{lem:cor-rot-st} Protocol \ref{proto:state_prep} is perfectly correct.
\end{lemma}

\begin{proof} The state of the central qubit after an honest execution of Protocol~\ref{proto:state_prep} before the correction sent by the Sender is $\ket{+_{\theta'}}$ with:
  \begin{equation}
    \label{eq:theta} \theta' = \theta_N + \sum\limits_{j \in [N-1]} (-1)^{t_j} \theta_j.
  \end{equation}
  It is sufficient to prove this for a pure state $\ket{\phi} = \alpha\ket{0} + \beta\ket{1}$ as control. Let $\ket{+_{\hat{\theta}}}$ with $\hat{\theta} \in \Theta$ be a rotated quantum state.  We apply a $\CNOT$ gate with $\ket{\phi}$ as control and $\ket{+_{\hat{\theta}}}$ as target, followed by a measurement of this second qubit in the computational basis.  Let~$t \in \bin$ be the measurement result. After tracing out the second qubit post-measurement, the system is in the following state:
  \begin{equation}
    \begin{split}
      \sqrt{2}\bra{0}_2 \X_2^t \CNOT_{12}\ket{\phi}\ket{+_{\hat{\theta}}}
      & = \bra{0}_2 \X_2^t \qty(\alpha \ket{00} + \alpha e^{i\hat{\theta}}\ket{01} + \beta \ket{11} + \beta e^{i\hat{\theta}} \ket{10}) \\
      & = \bra{0}_2(\alpha\ket{0} + \beta e^{i\hat{\theta}} \ket{1})\ket{t} + e^{i\hat{\theta}}\bra{0}_2(\alpha\ket{0} + \beta e^{-i\hat{\theta}} \ket{1})\ket{t \oplus 1} \\
      & = \Z(\hat{\theta})\ket{\phi}\braket{0}{t} + e^{i\hat{\theta}}\Z(-\hat{\theta})\ket{\phi}\braket{0}{t \oplus 1}
    \end{split}
  \end{equation}
  Therefore, the result of this single step is $\Z((-1)^t\hat{\theta})\ket{\phi}$ up to a global phase. Replacing the result above in the sequence of $\CNOT$'s and measurements performed by the Receiver where the control is qubit $N$ and the targets are qubits $j \neq N$ yields the desired value for $\theta'$.  Finally, the rotation correction $(-1)^b\theta - \theta'$ sent by the Sender, along with the $\X^b$, transform the value of the final state into $\ket{+_\theta}$.
\end{proof}

Furthermore, the computed angle $\theta'$ from Equation~\ref{eq:theta} is random from the point of view of the Receiver as long as at least a single angle is unknown.  We now show that this indeed provides perfect security to our protocol against an unbounded malicious Receiver.

\begin{lemma}[Security of Plugged Rotated Remote State Preparation]\label{lem:sec-rot-st} Protocol \ref{proto:state_prep} is perfectly secure against a malicious unbounded Receiver.
\end{lemma}

\begin{proof}
  We will construct here a Simulator against an adversarial Receiver, who expects to receive $N$ qubits, a number of leaks corresponding to the received states, and a final correction after transmitting the measurement results. The Simulator has single-query oracle access to the Stochastically Leaky Remote State Preparation Resource \ref{res:perf} for state set $\qty{\ket{+_\theta}}_{\theta \in \Theta}$ with leak probability $p_l^N$. It receives a state from this resource, potentially without the corresponding classical description, and must make the Receiver accept this state at the end of the interaction in both cases. The actions of this Simulator are described in Simulator \ref{sim:rot}.

 \begin{simulatorenv}[Malicious Receiver]
  \label{sim:rot}
  \item
  \begin{enumerate}
    \item It receives from the Receiver a string $\mathbf{c} \in \bin^N$ indicating whether the Receiver wishes to receive a leak from the Stochastically Leaky Random State Preparation Resources preparing the $N$ rotated states.
    \item If $\mathbf{c} = 1^N$ (i.e.\ the Receiver has requested leaks from all $N$ Resources $\mathcal{SL}_{p_l}$), it sends $c = 1$ to the Resource $\mathcal{SL}_{p_l^N}$. Otherwise it sends $c = 0$.
      \begin{itemize}
      \item If $c = 1$ and the leak is accepted by the Resource $\mathcal{SL}_{p_l^N}^C$:
        \begin{enumerate}
        \item The Simulator receives from the Resource a state $\ket{+_\theta}$ along with the corresponding angle $\theta$.
        \item It emulates the $N$ Resources $\mathcal{SL}_{p_l}$ by producing uniformly at random $N$ states $\ket{+_{\theta_j}}$ and leaking all the values $\theta_j$.
        \item It performs the rest of the protocol as an honest Sender would, using the value $\theta$ it received above, after which it halts.
        \end{enumerate}
      \item Otherwise, it simply receives a state $\ket{+_\theta}$ without the associated value of $\theta$ and continues the simulation below.
      \end{itemize}
    \item It then emulates the behaviour of the $N$ Resources $\mathcal{SL}_{p_l}$:
      \begin{enumerate}
      \item For each value $j$ such that $\mathbf{c}_j = 1$, the Simulator flips a biased coin with probability $p_l$ of getting $1$. If $\mathbf{c} = 1^N$, this step is repeated until at least a single coin does not return $1$ (since in this branch of the simulation not all states leak). Let $\mathbf{leak} \varsubsetneq [N]$ be the set of indices $j \in [N]$ such that $\mathbf{c}_j = 1$ and the coin-toss result is $1$.
      \item It samples uniformly at random a safe index $s \in_R [N] \setminus \mathbf{leak}$ and sends quantum states to the Receiver:
        \begin{itemize}
        \item For indices $j \neq s$ is samples uniformly at random $\theta_j \in_R \Theta$ and sends $\ket{+_{\theta_j}}$;
        \item For index $s$, it samples uniformly at random $\theta_s \in_R \Theta$ and $b_s \in_R \bin$, and sends an encrypted version $\Z(\theta_s)\X^{b_s}(\ket{+_\theta})$ of the state received from Resource $\mathcal{SL}_{p_l^N}^C$.
        \end{itemize}
      \item For all leaked indices $j \in \mathbf{leak}$, it sends $\theta_j$ to the Receiver.
      \end{enumerate}
    \item It receives from the Receiver a bit-string of measurement results $\mathbf{t} \in [N-1]$.
    \item After extending the bit-string $\mathbf{t}$ such that $\mathbf{t}_N = 0$, it computes $\theta'$ using Equation \ref{eq:theta} and sends the correction $(t_s \oplus b_s, -\theta')$ to the Receiver and halts.
    \end{enumerate}
  \end{simulatorenv}

  We can now prove that the view of any Distinguisher in the ideal case interacting with this Simulator is perfectly equal to its view in the real scenario when interacting with an honest Sender.

  The Simulator perfectly emulates the Stochastically Leaky Remote State Preparation Resources $\mathcal{SL}_{p_l}$ with exactly the same probability of leaking. In fact, the leak is accepted with probability $p_l^N$, which is exactly the probability that all $N$ Resources $\mathcal{SL}_{p_l}$ leak. Otherwise there exists necessarily an index for which the angle value has not been leaked and the Simulator is free to chose this position at random. Therefore this does not provide any distinguishing advantage.

  In the fully leaked case, the Simulator perfectly replicates the behaviour of the honest Sender since it has access to its input, therefore the real and ideal setting are perfectly indistinguishable. We henceforth focus on the case where at least one index has not leaked.

  Since the indistinguishability should hold regardless of which states are leaked, we can restrict our analysis to the worst case for the Simulator where all but one state have been leaked with the Receiver having requested all leaks. This gives the most power to the Distinguisher since a Distinguisher with more information can always chose to not use part of its knowledge to mimic Distinguishers with less information. We denote by $s$ the index of the state that has not leaked.

  In that setting, the Distinguisher knows the value of all $\theta_j$ for $j \neq s$ both in the real and ideal scenarii. Before sending the values for the measurement outcomes, the Distinguisher has in its possession -- not counting the leaked states for which it has perfect knowledge -- in the real case a state $\ket{+_{\theta_s}}$ and in the ideal case a state $\ket{+_{(-1)^{b_s}\theta + \theta_s}}$. This is indistinguishable since both are perfectly mixed states from the point of view of the Distinguisher as it does not know $\theta_s$.

  After sending the bit-string $\mathbf{t}$, regardless of how they were chosen, the Distinguisher receives a bit and an angle corresponding to the corrections chosen by the Sender or Simulator. In the first case this is equal to $(b, (-1)^b\theta - \theta')$ and in the second case $(t_s \oplus b_s, -\theta')$ with $b$ being chosen uniformly at random and $\theta'$ being computed in the exact same way in both settings. By knowing the leaked values of $\theta_j$ for $j \neq s$ and their associated $t_j$, we can rewrite these values to make them independent of the leaks. These angles and the corresponding states $\ket{+_{\theta_j}}$ are then independent of the other values and do not help the Distinguisher since their distribution is identical in both cases. The remaining parameters in the real case and ideal cases are respectively the chosen angle by the client, the received safe state, the corresponding measurement outcome in the gadget, the $\X$-correction bit and the $\Z$-correction angle:
  \begin{center}
    \begin{tabular}{r|rrrrrrr}
      Real-world  &$\Bigl($ & $\theta$; & $\ket{+_{\theta_s}}$; & $t_s$; & $b$; & $(-1)^b\theta - (-1)^{t_s}\theta_s$ & $\Bigr)$ \\
      Ideal-world &$\Bigl($ & $\theta$; & $\ket{+_{(-1)^{b_s}\theta + \theta_s}}$; & $t_s$; &$b_s \oplus t_s$; & $- (-1)^{t_s}\theta_s$ & $\Bigr)$
    \end{tabular}
  \end{center} We can now transform these by applying column-wide transformations, and here multiply the final angle by $(-1)^{t_s}$ and use this angle to apply a rotation to the state. This transforms the above parameter values into:
  \begin{center}
    \begin{tabular}{r|rrrrrrr}
      RW&$\Bigl($ & $\theta$; & $\ket{+_{(-1)^{b \oplus t_s}\theta}}$; & $t_s$; & $b$; & $(-1)^{b \oplus t_s}\theta - \theta_s$ & $\Bigr)$ \\
      IW&$\Bigl($ & $\theta$; & $\ket{+_{(-1)^{b_s}\theta}}$; & $t_s$; &$b_s \oplus t_s$; & $- \theta_s$ & $\Bigr)$
    \end{tabular}.
  \end{center}
  Note that all these rewriting operations are reversible and can be done by the Distinguisher. The amount of information that may be deduced from the original real or ideal execution is therefore equal to that which may be extracted from these rewritten values. Finally we see that in both cases, the value for $\theta_s$ only appears in the final term, and since it is chosen uniformly at random both final terms follow the same distribution, meaning that they give no distinguishing advantage. They can therefore be omitted from the distribution and we get:
  \begin{center}
    \begin{tabular}{r|rrrrrrr}
      RW&$\Bigl($ & $\theta$; & $\ket{+_{(-1)^{b \oplus t_s}\theta}}$; & $t_s$; & $b$; & $\Bigr)$ \\
      IW&$\Bigl($ & $\theta$; & $\ket{+_{(-1)^{b_s}\theta}}$; & $t_s$; &$b_s \oplus t_s$; & $\Bigr)$
    \end{tabular}
  \end{center}
  Since $b$ in the first distribution is a bit sampled uniformly at random, we can substitute by $b\oplus t_s$ without changing the distribution. We finally arrive at
  \begin{center}
    \begin{tabular}{r|rrrrrrr}
      RW&$\Bigl($ & $\theta$; & $\ket{+_{(-1)^{b}\theta}}$; & $t_s$; & $b \oplus t_s$; & $\Bigr)$ \\
      IW&$\Bigl($ & $\theta$; & $\ket{+_{(-1)^{b_s}\theta}}$; & $t_s$; &$b_s \oplus t_s$; & $\Bigr)$
    \end{tabular}
  \end{center} These two distributions are clearly identical given that both $b$ and $b_s$ are uniformly random bits, which concludes the proof.
\end{proof}

We now obtain the following corollary by combining Theorem \ref{thm:boost} and Lemma \ref{lem:stoc-to-perf} to first construct a good Stochastically Leaky Remote State Preparation Resource from $N$ bad ones, and then to construct from the good one a Leak Free Remote State Preparation Resource for rotated states.

\begin{corollary}
  \label{cor:perf-st}
  Given the set of states $\Phi = \qty{\ket{+_\theta}}_{\theta \in \Theta}$, Protocol \ref{proto:state_prep} $\epsilon$-constructs in the AC Framework an Ideal Remote State Preparation Resource for the set $\Phi$ using $N$ Stochastically Leaky Remote State Preparation Resources $\mathcal{SL}_{p_l}$ for the set $\Phi$, with $\epsilon = \frac{7}{8}p_l^N$.
\end{corollary}

The result above shows that any device that leaks information about these rotated states with constant probability $p_l$ is equivalent to a device that is exponentially close to perfect using only a linear number of repetitions and simple operations on the Receiver's side.
 
\subsection{Suppressing Faults via Verifier-side Gate Compilation}
\label{sec:rsp-ac}
Our goal is to construct an RSP resource using Stochastically Compromised Preparations, Gates and Measurements (Resource~\ref{res:scpg}). At a high-level it will rely on preparing a $\ket +$ state and applying $\Z(\theta)$ to it. Yet, as Resource~\ref{res:scpg} is possibly faulty, we will need to encode the quantum information into a quantum error correcting code and to require some additional properties, akin to those of fault-tolerance, to be satisfied in order to ensure both correctness and security of the constructed RSP resource.

To this end:
\begin{enumerate}
  \item In \S~\ref{subsubsec:def_targ}, we define the target RSP as a Level-$k$ RSP (Resource~\ref{res:k-rsp}) and show that is it secure by design, thereby able to replace the initial RSP (Resource~\ref{res:rsp}) when dealing with encoded computations;
  \item In \S~\ref{subsubsec:recursive_sim}, we recursively define simulations for preparations, gates and measurements with increasing concatenation levels of the $[[15,1,3]]$ quantum Reed-Muller code using Stochastically Compromised Preparations, Gates and Measurements;
  \item In \S~\ref{subsubsec:safe_properties}, we give properties of these simulations which will make them safe to use in RSP protocols which suffer from faults and leaks;
  \item In \S~\ref{subsubsec:threshold_thm}, we show that, for a compromise probability $p_c < p_0$, as the level-$k$ of encoding increases, these simulations are \emph{accurate} and \emph{private} with probability doubly exponentially close to 1 --- i.e. they perform the expected transformation (accuracy) on the logical qubit while possibly affecting the syndrome qubits in a way that is independent of the chosen value of $\theta$ parametrizing $\Z(\theta)$ gates (privacy). This can in turn be combined into a threshold theorem for accurate and private computations with compromised gates;
  \item In \S~\ref{subsubsec:level_k_rsp}, we use the previous theorem to construct a level-$k$ RSP with negligible error with respect to $k$.
\end{enumerate}

The reader familiar with the threshold theorem for fault-tolerance will note that the steps outlined above and expanded below follow quite closely~\cite{AGP06quantum}. This is no coincidence as our work indeed shows how to perform secure delegated fault-tolerant quantum computations. Yet, because faults can be converted into leaks when errors are commuted through classically controlled gates, it was necessary to reprove correctness --- i.e. standard fault-tolerance --- in a slightly different language so that it is also adapted to the requirements of the AC framework. Except this, we tried to minimize the unnecessary differences to make it easier to see where security forces to depart from standard fault-tolerance.

\subsubsection{Definition of the Target Resource}
\label{subsubsec:def_targ}

As explained earlier, our goal is to produce logically encoded $\ket{+_{\theta}}$ states that can then be sent to the Prover in a protocol for verifying a quantum computation. We have also shown that the problem with logically encoded information is that the syndrome can carry away information about the secret value $\theta$. These two facts motivate the following target RSP resource: 

\begin{resourceenv}[Level-$k$ Remote State Preparation]
  \label{res:k-rsp}
  \item
  \begin{algorithmic}[0]

    \STATE \textbf{Public Information:} The set $\Theta = \{j\pi/4. \  j = 0, \ldots 7\}$, and an integer $k$ representing the concatenation level of the $[[15,1,3]]$ quantum Reed-Muller code.

    \STATE \textbf{Sender's Input:} A value $\theta \in \Theta$.

    \STATE \textbf{Receiver's Input:} A bit $c \in \{0, 1\}$ indicating whether it wants to cheat, set to $0$ in the honest case.

    \STATE \textbf{Computation by the Resource:}
    \begin{itemize}
    \item If $c = 0$ outputs the logical $\ket {+_{\theta}}$ state of the concatenated code at the Receiver's interface.
    \item If $c = 1$:
      \begin{itemize}
      \item Gets a quantum state $\sigma$ of $15^k-1$ qubits from the Receiver.
      \item Performs the encoding operation into the code using the input state $\ketbra{+_{\theta}}\otimes \sigma$.
      \item Outputs this state to the Receiver's interface.
      \end{itemize}
    \end{itemize}
\end{algorithmic}
\end{resourceenv}

This resource clearly fulfils our expectations in terms of correctness and security. The output when there is no cheating is the desired $\ket{+_{\theta}}$ state with a syndrome that does not carry any information about $\theta$ because it is not correlated with it. Whenever the Receiver cheats, it does not have more about $\theta$ than it would have from receiving $\ket{+_{\theta}}$. Yet, it can still chose the value of the syndrome that it receives from it, allowing to model situations where noise induces non trivial syndromes, albeit being again uncorrelated with $\theta$.

\subsubsection{Recursive Simulation Protocols}
\label{subsubsec:recursive_sim}

\renewcommand{\arraystretch}{1.25}
\begin{table}[ht]
\centering
\begin{tabular}{p{4.5cm}p{4cm}p{6.5cm}}
\toprule
\textbf{Building block in standard fault-tolerance} & \textbf{Analogue in secure fault-tolerance} & \textbf{Description} \\
\midrule
Error correction ($\mathsf{EC}$) & ---\emph{same building block}--- \\
Gate ($\mathsf{Ga}$) & Safe gate ($\mathsf{SafeGa}$) & In the case of unparametrised gates, $\mathsf{SafeGa}$ is equal to $\mathsf{Ga}$. The difference arises only in the case of parametrised gates with secret parameters. In the proposed protocol, the only different $\mathsf{SafeGa}$'s are safe $\Z$-rotations (see Def.~\ref{def:safe-z-theta}). \\
Rectangle ($\mathsf{Rec}$) & Safe rectangle ($\mathsf{SafeRec}$) & $\mathsf{SafeRec}$'s are recursively constructed from $\mathsf{SafeGa}$'s and $\mathsf{EC}$'s (see Def.~\ref{def:saferec}). \\
Extended rectangle \newline ($\mathsf{exRec}$) & Extended safe rectangle ($\mathsf{exSafeRec}$) & $\mathsf{exSafeRec}$'s consist of $\mathsf{SafeRec}$'s and directly preceding $\mathsf{EC}$'s (see Def.~\ref{def:exsaferec}), and are constructed to be overlapping with neighboring $\mathsf{exSafeRec}$'s. \\
\bottomrule
\end{tabular}
\caption{An overview of the building blocks of Recursive Simulation Protocols, in standard fault-tolerance and \emph{secure} fault-tolerance. All building blocks are recursively defined through different levels; a level-$k$ version of a building block is denoted by adding a "$k-$" prefix.}
\label{table:building-blocks-recursive-protocols}
\end{table}

To proceed with the definition of Recursive Simulation Protocols, we first define Safe Rectangles ($\mathsf{SafeRec}$), which are the analogues of Rectangles ($\mathsf{Rec}$) in standard fault-tolerance. Similarly, we proceed with the definition of extended $\mathsf{SafeRec}$ ($\mathsf{exSafeRec}$) that match $\mathsf{exRec}$. The $\mathsf{Safe}$ qualifier will be justified later when we will prove that these not only provide better accuracy but also better privacy when they are concatenated. Following~\cite{AGP06quantum}, we provide the explicit definitions for gates. Those for preparations and measurements can easily be recovered. Recall that all level-$0$ operations are performed using Stochastically Compromised Preparations, Gates, and Measurements described in Resource~\ref{res:scpg}.

\begin{definition}[$1-\mathsf{SafeRec}$]
  A $1-\mathsf{SafeRec}$ for simulating a $0-\mathsf{Ga}$ consists of the $1-\mathsf{Ga}$ followed by two consecutive $1-\mathsf{EC}$ for each output block of the $1-\mathsf{Ga}$.
\end{definition}

\begin{definition}[$k-\mathsf{SafeRec}$]\label{def:saferec}
  A $k-\mathsf{SafeRec}$ is obtained from a $(k-1)-\mathsf{SafeRec}$ by replacing each $0-\mathsf{Ga}$ by its corresponding $1-\mathsf{SafeRec}$.
\end{definition}

\begin{definition}[$k-\mathsf{exSafeRec}$]\label{def:exsaferec}
  A \emph{$k-\mathsf{exSafeRec}$} is a $k-\mathsf{SafeRec}$ augmented with the $k-\mathsf{EC}$s from the $k-\mathsf{SafeRec}$ that precede it.\footnote{Note that a state preparation does not have a $k-\mathsf{EC}$ preceding it, therefore the $k-\mathsf{exSafeRec}$ of a preparation step is simply the $k-\mathsf{SafeRec}$.}
\end{definition}

Figure~\ref{fig:exsaferec} depicts two $1-\mathsf{SafeRec}$'s and the composition of one $1-\mathsf{exSafeRec}$.

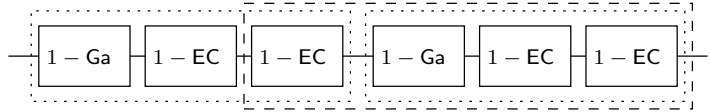
\begin{figure}
  \centering
  \begin{picture}(0,0)\includegraphics{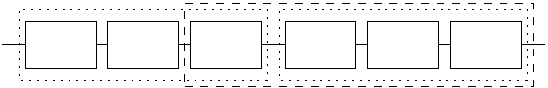}\end{picture}\setlength{\unitlength}{4144sp}\begin{picture}(4164,654)(1249,332)
\put(1486,614){\makebox(0,0)[lb]{\smash{\fontsize{8}{9.6}\usefont{T1}{ptm}{m}{n}{\color[rgb]{0,0,0}$1-\mathsf{Ga}$}}}}
\put(2116,614){\makebox(0,0)[lb]{\smash{\fontsize{8}{9.6}\usefont{T1}{ptm}{m}{n}{\color[rgb]{0,0,0}$1-\mathsf{EC}$}}}}
\put(2746,614){\makebox(0,0)[lb]{\smash{\fontsize{8}{9.6}\usefont{T1}{ptm}{m}{n}{\color[rgb]{0,0,0}$1-\mathsf{EC}$}}}}
\put(3466,614){\makebox(0,0)[lb]{\smash{\fontsize{8}{9.6}\usefont{T1}{ptm}{m}{n}{\color[rgb]{0,0,0}$1-\mathsf{Ga}$}}}}
\put(4096,614){\makebox(0,0)[lb]{\smash{\fontsize{8}{9.6}\usefont{T1}{ptm}{m}{n}{\color[rgb]{0,0,0}$1-\mathsf{EC}$}}}}
\put(4726,614){\makebox(0,0)[lb]{\smash{\fontsize{8}{9.6}\usefont{T1}{ptm}{m}{n}{\color[rgb]{0,0,0}$1-\mathsf{EC}$}}}}
\end{picture}   \caption{Short-dashed rectangles correspond to a $1-\mathsf{SafeRec}$, while the long-dashed one corresponds to a $1-\mathsf{exSafeRec}$.}
  \label{fig:exsaferec}
\end{figure}

We also give a new (safe) implementation of the $1-\mathsf{Ga}$ for the single qubit $\Z(\theta)$ gate for the $[[15,1,3]]$ quantum Reed Muller code. The $\Z(\theta)$ gate is transversal for this code and is usually implemented by applying the physical rotation $\Z(\theta)$ to all the physical qubits of the encoded qubit. However, this means that a single leak at any location would reveal the full secret value of $\theta$. After a level-$k$ concatenation, the probability of leaking over any one of the $15^k$ locations grows exponentially in $k$. This new implementation will allow us prevent this leakage.

\begin{definition}[$1-\mathsf{Safe}\Z(\theta)$]\label{def:safe-z-theta}
  Using the $[[15,1,3]]$ quantum Reed Muller code, the $1-\mathsf{Safe}\Z(\theta)$ gate that implements the $1-\mathsf{Ga}$ for the single qubit $\Z(\theta)$ gate is constructed as follows:
  \begin{algorithmic}[0]
    \STATE \textbf{For $i \in [1,15]$}:

    \begin{itemize}
    \item Sample $\alpha_i \sample \mathcal{U}(\Theta)$, where $\mathcal{U}(\Theta)$ is the uniform distribution over $\Theta$.
    \item Set $\beta_i = \theta - \alpha_i$.
    \item Apply the two $0-\mathsf{Ga}$'s for $\Z(\beta_i) \circ \Z(\alpha_i)$ to the $i$\textsuperscript{th} qubit, via calls to the Stochastically Compromised Gates Resource.
    \end{itemize}
  \end{algorithmic}
  
\end{definition}

Figure~\ref{fig:logical_Z} depicts a $2-\mathsf{SafeRec}$ implementing the $\Z(\theta)$ gate.

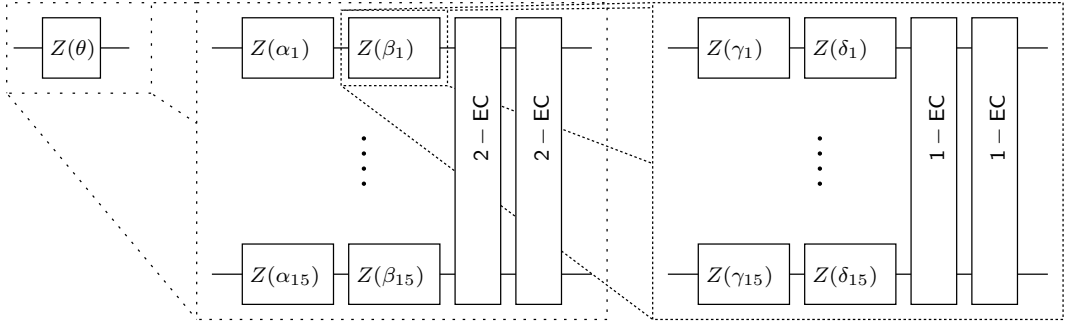
\begin{figure}
  \centering
  \begin{picture}(0,0)\includegraphics{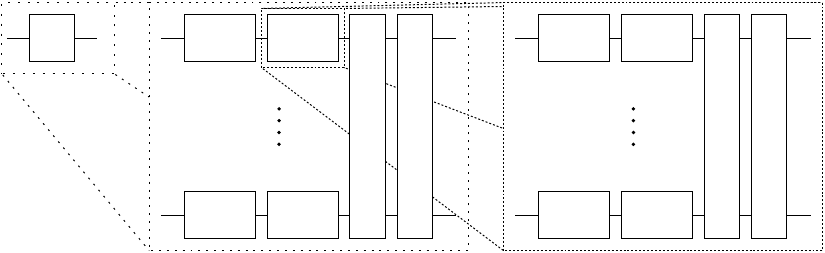}\end{picture}\setlength{\unitlength}{4144sp}\begin{picture}(6279,1914)(34,-973)
\put(301,614){\makebox(0,0)[lb]{\smash{\fontsize{8}{9.6}\usefont{T1}{ptm}{m}{n}{\color[rgb]{0,0,0}$Z(\theta)$}}}}
\put(1486,614){\makebox(0,0)[lb]{\smash{\fontsize{8}{9.6}\usefont{T1}{ptm}{m}{n}{\color[rgb]{0,0,0}$Z(\alpha_{1})$}}}}
\put(1486,-736){\makebox(0,0)[lb]{\smash{\fontsize{8}{9.6}\usefont{T1}{ptm}{m}{n}{\color[rgb]{0,0,0}$Z(\alpha_{15})$}}}}
\put(2116,-736){\makebox(0,0)[lb]{\smash{\fontsize{8}{9.6}\usefont{T1}{ptm}{m}{n}{\color[rgb]{0,0,0}$Z(\beta_{15})$}}}}
\put(2116,614){\makebox(0,0)[lb]{\smash{\fontsize{8}{9.6}\usefont{T1}{ptm}{m}{n}{\color[rgb]{0,0,0}$Z(\beta_{1})$}}}}
\put(2896,-16){\rotatebox{90}{\makebox(0,0)[lb]{\smash{\fontsize{8}{9.6}\usefont{T1}{ptm}{m}{n}{\color[rgb]{0,0,0}$\mathsf{2-EC}$}}}}}
\put(3256,-16){\rotatebox{90}{\makebox(0,0)[lb]{\smash{\fontsize{8}{9.6}\usefont{T1}{ptm}{m}{n}{\color[rgb]{0,0,0}$\mathsf{2-EC}$}}}}}
\put(4186,614){\makebox(0,0)[lb]{\smash{\fontsize{8}{9.6}\usefont{T1}{ptm}{m}{n}{\color[rgb]{0,0,0}$Z(\gamma_{1})$}}}}
\put(4186,-736){\makebox(0,0)[lb]{\smash{\fontsize{8}{9.6}\usefont{T1}{ptm}{m}{n}{\color[rgb]{0,0,0}$Z(\gamma_{15})$}}}}
\put(4816,-736){\makebox(0,0)[lb]{\smash{\fontsize{8}{9.6}\usefont{T1}{ptm}{m}{n}{\color[rgb]{0,0,0}$Z(\delta_{15})$}}}}
\put(4816,614){\makebox(0,0)[lb]{\smash{\fontsize{8}{9.6}\usefont{T1}{ptm}{m}{n}{\color[rgb]{0,0,0}$Z(\delta_{1})$}}}}
\put(5596,-16){\rotatebox{90}{\makebox(0,0)[lb]{\smash{\fontsize{8}{9.6}\usefont{T1}{ptm}{m}{n}{\color[rgb]{0,0,0}$\mathsf{1-EC}$}}}}}
\put(5956,-16){\rotatebox{90}{\makebox(0,0)[lb]{\smash{\fontsize{8}{9.6}\usefont{T1}{ptm}{m}{n}{\color[rgb]{0,0,0}$\mathsf{1-EC}$}}}}}
\end{picture}   \caption{Safe Level-$2$ Simulation of $\Z(\theta)$ Gate. Left: Level-2 (Logical level); middle: level-$1$ simulation with $\alpha_i$ chosen at random and $\theta = \alpha_i + \beta_i$; right: level-$0$ simulation depicting the implementation of the level-$1$ $\Z(\beta_{1})$, with $\gamma_i$ chosen at random and $\beta_1 = \gamma_i + \delta_i$. To leak the value of $\theta$, it is necessary to leak all the angles for at least one physical qubit which decreases doubly exponentially with the simulation level.}
  \label{fig:logical_Z}
\end{figure}

\subsubsection{Properties of \texorpdfstring{$\mathsf{Safe}$}{Safe} Constructions}
\label{subsubsec:safe_properties}

We first give a definition which captures what our $\mathsf{Safe}$ gadgets defined above are expected to achieve.

\begin{definition}[Accurate and Private $k-\mathsf{SafeRec}$]
\label{def:cor-sec}

    A $k-\mathsf{SafeRec}$ implementing a family of gates parametrized by a classical parameter $\lambda \in \Lambda$, with all gates in the family having the same set of compromised locations, is \emph{accurate and private} if
    \begin{itemize}
    \item \emph{(Accuracy)} The transformation that corresponds to applying the $k-\mathsf{SafeRec}$ followed by the ideal full $k$-decoder and tracing out the syndrome register is equivalent to that obtained by applying the ideal full $k$-decoder, tracing out the syndrome register and applying the ideal $0-\mathsf{Ga}$ to the remaining logical register;
    \item \emph{(Privacy)} The state of the syndrome register when applying the $k-\mathsf{SafeRec}$ followed by the ideal full $k$-decoder is independent of $\lambda$.
    \end{itemize}
Above, an ideal full $k$-decoder is a unitary transformation that is mapping a $k$-block into a single logical qubit and a syndrome register. Because of the concatenated structure of the code, an ideal full $(k+1)$-decoder can be realized by having ideal full $k$-decoders being applied to $k$-blocks followed by an ideal full $1$-decoder acting on the $1$-block constructed with the logical outputs of the ideal full $k$-decoders. These full $k$-decoders are introduced as $k$-*decoders in \cite{AGP06quantum}.
\end{definition}

Any deviations from the ideal properties of accuracy and privacy in Definition~\ref{def:cor-sec} are measured in the metric given by the trace distance (for states) and the diamond norm (for channels).
This is consistent with the choice of the metric in the Abstract Cryptography framework, and behaves well under composition.

Intuitively, the first property guarantees that the $k-\mathsf{SafeRec}$ performs the expected unitary transformation on the logical register, and the impact of errors is located on the syndrome register only. Then, because the state of the syndrome is independent of $\lambda$, the second property ensures that the deviation imposed by the compromised gates on the $k$-blocks at the output of the $k-\mathsf{SafeRec}$ is independent of $\theta$. Therefore, it could be entirely simulated by the Eavesdropper if it is given the state produced by the perfect $k-\mathsf{SafeRec}$ (i.e. without cheating).

  Note that the second property is not satisfied by the (unsafe) canonical transversal implementation of the $\Z(\theta)$ gate for the $[[15,1,3]]$ Reed Muller code. There, a bit flip on e.g.~the first qubit would yield a syndrome that is dependent on the value of $\theta$, as was shown in \S~\ref{sec:modeling}.
  
  We now give recursive three properties of $k-\mathsf{exSafeRec}$s, which will be categorised by the number of errors in their lower-level components.

\begin{definition}[Good, Bad and Independent $k-\mathsf{exSafeRec}$s]
\label{def:good_bad_indep}
    
    At level-$1$, a $1-\mathsf{exSafeRec}$ is \emph{good} if it contains no more than one compromised location. If it is not good, it is \emph{bad}.
    
    Two bad $k-\mathsf{exSafeRec}$ are \emph{independent} if they are not overlapping or if they overlap and the earlier $k-\mathsf{exSafeRec}$ is still bad when the shared $k-\mathsf{EC}$ is removed.
    
    At level-$k$, a $k-\mathsf{exSafeRec}$ is \emph{bad} if it contains two or more independent bad $(k-1)-\mathsf{SafeRec}$. If it is not bad it is \emph{good}.
    
\end{definition}
  
  We will now show that a construction which satisfies the property from Definition~\ref{def:good_bad_indep} will also be accurate and private as per Definition~\ref{def:cor-sec}. This in turn will allow us to prove a threshold theorem for accuracy and privacy, showing that increasing the level of concatenation is beneficial.

\subsubsection{Concatenation Increases Accuracy and Privacy}
\label{subsubsec:threshold_thm}

We present here the main technical contribution of the subsection. It shows that there exists a threshold $p_0$ for the compromise probability $p_c$ below which, whenever a computation uses $k-\mathsf{SafeRec}$s and the $k-\mathsf{Safe}\Z(\theta)$ gate, it can be made accurate and private with probability doubly exponentially close to $1$ as $k$ increases. Here, accuracy and privacy correspond to Definition~\ref{def:cor-sec} extended to the output of a full computation. To proceed with the proof, we will go through the following steps:

  \begin{itemize}
  \item Examine how to obtain an accurate and private $1-\mathsf{SafeRec}$ (Lemma~\ref{lem:good_correct_level_1});
  \item Proceed recursively to obtain a condition for accurate and private $k-\mathsf{SafeRec}$ (Lemma~\ref{lem:good_correct_level_k});
  \item Show show that it will be met with high probability as $k$ increases (Lemma~\ref{lem:rare_bad_recs});
  \item Combine these results to obtain the threshold for accurate and private computation with Stochastically Compromised Preparations, Gates and Measurements (Theorem~\ref{thm:safe_threshold}).
  \end{itemize}
  
Our target being RSP, we will only consider the $\Z(\theta)$ gate and its level-$k$ simulation $k-\mathsf{Safe}\Z(\theta)$. We define good, bad and independent $\mathsf{exSafeRec}$s and link these properties to the accurate and private properties of $\mathsf{SafeRec}$s defined earlier.

  \begin{lemma}[Good implies accurate and private at level-$1$]
  \label{lem:good_correct_level_1}
  
    Suppose all level-$1$ gadgets obey properties P0-P4 from \S~\ref{sec:ftqc}. Then, if a $1-\mathsf{exSafeRec}$ is good and the input to its $1-\mathsf{SafeRec}$ has no more than one error in each input block, it is accurate and private.

\end{lemma}

  \begin{proof}

    Suppose that none of the leading $1-\mathsf{EC}$ of the $1-\mathsf{exSafeRec}$ has a compromised location. Then the output of each $1-\mathsf{EC}$ is a codeword (P0). Since the input to the $1-\mathsf{SafeRec}$ has at most one error, this means that there are actually no error to the $1-\mathsf{SafeRec}$. The compromised location can be either inside the $1-\mathsf{Ga}$ or in one of the two series of trailing $1-\mathsf{EC}$s. If the compromised location is not in the second trailing $1-\mathsf{EC}$, then this last $1-\mathsf{EC}$ will correct the error and produce the all-zero syndrome. If the compromised location is on the last $1-\mathsf{EC}$, then because the $1-\mathsf{Ga}$ and the preceding $1-\mathsf{EC}$ contain no compromised location, they will have reset the syndrome to $0$ and the produced syndrome at the output of the last $1-\mathsf{EC}$ will be uncorrelated with $\theta$.
    Finally, we note that when the $1-\mathsf{Ga}$ is the $1-\mathsf{Safe}\Z(\theta)$, a compromised location in one of the physical $\Z(\theta)$ gate discloses some value drawn at random from $\mathcal{U}(\Theta)$ which is uncorrelated with $\theta$. Hence, we can conclude that whenever the compromised location is in the $1-\mathsf{SafeRec}$, it is accurate and private.

    Now we turn to the case where the compromised location is in one of the leading $1-\mathsf{EC}$ block of the $1-\mathsf{exSafeRec}$. Then, we can conclude that at most one input block has an error. Because neither the following $1-\mathsf{Ga}$, nor any of the two consecutive $1-\mathsf{EC}$s has a compromised location, then the error will be corrected and the syndrome will be back to the all-zero syndrome. As a consequence, the $1-\mathsf{SafeRec}$ is accurate and private.

    The arguments presented above regarding accuracy also apply to all the other gates, preparations and measurements. Moreover, because these are not parametrized, then privacy is trivial and the statement of the lemma holds for any preparation, gate and measurement.

    Additionally, note that the accuracy and privacy hold for a truncated $1-\mathsf{SafeRec}$ where the second trailing $1-\mathsf{EC}$s are removed. Accuracy is obvious, while the reason for privacy is that either the first trailing $1-\mathsf{EC}$ has no compromised location, then it resets the syndrome, or that it has a compromised location, but then there are none elsewhere before so that the value of $\theta$ is not disclosed and the syndrome before entering the $1-\mathsf{EC}$ with the compromised location is all zero so that the output of the $1-\mathsf{EC}$ produces a syndrome that is independent of $\theta$.
  \end{proof}

  Having characterized that good implies accurate and private at level-$1$, we can now proceed with the induction step, showing that good implies accurate and private at level-$k$.

  \begin{lemma}[Good implies accurate and private]
  \label{lem:good_correct_level_k}  
  
    A good $k-\mathsf{exSafeRec}$ implies that its $k-\mathsf{SafeRec}$ is accurate and private.
    
  \end{lemma}

  \begin{proof}
  The proof of this lemma will be done by induction. We will assume the property holds up until level-$k$. We will then take a $(k+1)-\mathsf{exSafeRec}$ apply an ideal full $k$-decoder at the output of the $(k+1)$-blocks. We then show that this is equivalent to running the ideal full $(k+1)$-decoder before the $(k+1)-\mathsf{SafeRec}$, followed by the ideal gate acting on the logical part tensored by a CPTP acting on the syndrome in a way that is independent of $\theta$; and that by doing so the value of $\theta$ has not been disclosed.

  To this end, we will rely on the recursive structure of the ideal full $(k+1)$-decoder that consists of performing ideal full $k$-decoding followed by a full $1$-decoder. We will then sweep the $k$-decoders through the $k-\mathsf{SafeRec}$ transforming the $(k+1)-\mathsf{SafeRec}$ into a $1-\mathsf{SafeRec}$ by applying the induction hypothesis. Then we will conclude using the property at the level-$1$. As for standard fault-tolerance \cite{AGP06quantum}, the difficulty of this program is that the induction hypothesis does not apply to bad $k-\mathsf{exSafeRec}$s, so that we will have to show that such situation can be accurately simulated using a compromised $0-\mathsf{Ga}$.
    
\item \emph{Non-independent pairs of bad $k-\mathsf{exSafeRec}$s.} By construction, the two bad $k-\mathsf{exSafeRec}$s have a common $k-\mathsf{EC}$ so that when removed from the first $k-\mathsf{exSafeRec}$, this $k-\mathsf{exSafeRec}$ is not bad. We can then sweep the ideal full decoder past the last $k-\mathsf{exSafeRec}$ (entirely, not just the $k-\mathsf{SafeRec}$) which happens to be bad. The next paragraph will show that this translates into a compromised $0-\mathsf{Ga}$. Now, we can apply the induction hypothesis to the truncated $k-\mathsf{SafeRec}$ to conclude that this one yields a accurate and private simulation of the corresponding ideal $0-\mathsf{Ga}$.
  
\item \emph{Transforming a bad $k-\mathsf{exSafeRec}$ into a compromised $0-\mathsf{Ga}$.} Here, we are faced with the same difficulty as in the original FTQC analysis. The induction hypothesis does not state that a bad $k-\mathsf{exSafeRec}$ can be simulated by a compromised $0-\mathsf{Ga}$. Yet, we are in a slightly different situation as our definition of accurate and private directly requires the use of ideal full $k$-decoders.  The consequence is that our induction hypothesis is indeed stronger and directly gives the shape of the transformation acting on the logical and syndrome registers.

  Indeed, when ideal full $k$-decoders sweep past the $k-\mathsf{SafeRec}$ inside a good $k-\mathsf{exSafeRec}$, our induction hypothesis states that we can replace it by the ideal full $k$-decoder followed by a unitary transformation that applies $\cptp U(\theta)\otimes \mathcal M$ over the logical and syndrome registers with the added condition that $\mathcal M$ is independent of $\theta$. This also holds for truncated $k-\mathsf{SafeRec}$s.

  For a bad $k-\mathsf{exSafeRec}$, this is not true any more. If $\mathcal D$ is the unitary transformation corresponding to the ideal full $k$-decoder and $\mathcal N$ the CPTP that the Eavesdropper imposes through the compromised locations, then the action on the logical and syndrome registers is $\mathcal D \mathcal N \mathcal D^{-1}$. This can be regarded as a level-$0$ compromised gate simulated by the bad $k-\mathsf{exSafeRec}$. To see this, we realize that the compromised $0-\mathsf{Ga}$ will send the logical qubits to the Eavesdropper together with the classical control value that parametrizes the $0-\mathsf{Ga}$. This allows the Eavesdropper to simulate any transformation that the bad $k-\mathsf{exSafeRec}$ implements. Hence, it does not restrict its adversarial power. As a consequence, the Eavesdropper has the ability to perform highly correlated attacks at various compromised locations. This was indeed already taken into account in the definition of good and bad $1-\mathsf{exRec}$ in \cite{AGP06quantum} as these notions depend only on the number of compromised locations, not on the actual CPTP map induced by the Eavesdropper on the corresponding registers.

\item \emph{From level-$k$ to level-$(k+1)$.} Consider an ideal full $(k+1)$-decoder following a $(k+1)-\mathsf{exSafeRec}$. It can be realized by ideal full $k$-decoders followed by ideal full $1$-decoders. When the $k$-decoders sweep to the front of the $(k+1)-\mathsf{exSafeRec}$, each good $k-\mathsf{exSafeRec}$ is transformed into an accurate and private $0-\mathsf{Ga}$, and only bad $k-\mathsf{exSafeRec}$s yield to compromised $0-\mathsf{Ga}$ that leak the value of their parameter $\theta$ and share access to the syndrome of the corresponding $k$-block. This means that the obtained $1-\mathsf{exSafeRec}$ is good and that the $1-\mathsf{Rec}$ inside it is accurate. As a result, we can move the $1$-decoder to the left of the $1-\mathsf{SafeRec}$ transforming it into an accurate and private $0-\mathsf{Ga}$. Now rejoining the $k$- and $1$-decoders we reconstruct the ideal full $(k+1)$-decoders. 

  This concludes the proof as it shows that the operation induced by the $(k+1)-\mathsf{SafeRec}$ on the logical and syndrome registers is the ideal gate tensored with a CPTP map that does not depend on $\theta$.    

  This generalizes to preparations and measurements.
  \end{proof}
  
  Finally, we show that the probability of a $\mathsf{k-exSafeRec}$ being bad is negligible in $k$.
  
    \begin{lemma}[Bad $\mathsf{k-exSafeRec}$s are rare]
	\label{lem:rare_bad_recs}    
    
    Suppose that compromised locations are drawn independently at random with probability $p_c$ at each circuit location in a $k-\mathsf{exSafeRec}$. Then, there exists a $p_0 > 0$ such that the probability $p^{(k)}_c$ that the $k-\mathsf{exSafeRec}$ is bad is
    \begin{equation}
      p^{(k)}_c \leq p_0 (p_c/p_0)^{2^k}.
    \end{equation}
    Hence, for $p_c < p_0$, \emph{i.e.}, below the threshold, $p^{(k)}_c$ decreases doubly exponentially in $k$.
  \end{lemma}

  \begin{proof}
    Let $A$ be the number of pairs of possible compromised locations in the largest $1-\mathsf{exSafeRec}$.  Then, the probability of having two independent compromised locations in a $1-\mathsf{exSafeRec}$ satisfies $p^{(1)}_c \leq A p_c^{2}$. Because of the recursive nature of the construction of the $k-\mathsf{exSafeRec}$s, we have $p^{(k)}_c \leq A \left(p^{(k-1)}_c\right)^{2}$. So that we obtain
    \begin{equation}
      p^{(k)}_c \leq p_0 (p_c/p_0)^{2^k},
    \end{equation}
    for $p_0 = A^{-1}$.
  \end{proof}

  \begin{theorem}[Threshold theorem for accurate and private computations]
  \label{thm:safe_threshold}
  
    Let $L$ be the number of locations in an unencoded circuit $C$. The diamond distance $\delta$ between the circuit $C$ simulated with level-$k$ and an ideal computation performing the expected computation $C$ on the logical space and producing a syndrome that is independent of the parameters $\theta$ of the gates satisfies:
    \begin{equation}
      \delta \leq 2L \times p_0(p_c/p_0)^{2^k},
    \end{equation}
    where $p_0^{-1}$ counts the number of pairs of possible compromised locations in the largest $1-\mathsf{exSafeRec}$, and $p_c$ is the probability that a location is compromised. As a consequence, for $p_c < p_0$, \emph{i.e.}, below the threshold, increasing $k$ yields a doubly exponential reduction of $\delta$.
  \end{theorem}

  \begin{proof}
    To ensure an accurate and private simulation, we require that no $k-\mathsf{exSafeRec}$ is bad, for it could not guarantee that its logical qubits are accurately produced nor that the syndrome perfectly hides the classical parameters of the gates that it was supposed to conceal and simulate. On the contrary, whenever a $k-\mathsf{exSafeRec}$ is good, we have shown that the $k-\mathsf{SafeRec}$ inside it produces an output that can be written as the expected operation on the logical qubits and a CPTP map on the syndrome qubits that is independent of the parameter that defines the gate.

    As a consequence the probability of failure is bounded by $L$ times the individual failure probability per location --- i.e.~$L \times p_0 (p_c/p_0)^{2^k}$. This yields the stated result because the diamond norm is bounded by $2$.
  \end{proof}

\subsubsection{Construction of a \texorpdfstring{Level-$k$}{Level-k} Remote State Preparation}
\label{subsubsec:level_k_rsp}

Finally, we present the protocol for implementing the Level-$k$ Remote State Preparation, which aims to correct maliciously chosen errors and prevent leaks.

\begin{protocolenv}[Level-$k$ Simulation of Remote State Preparation]\label{proto:ft-rsp}
\item 
  \begin{algorithmic}[0]
    \STATE \textbf{The Sender:} 
	\begin{itemize}
	\item Chooses $\theta \in \Theta$;
	\item Implements the level-$k$ preparation of $\ket +$;
	\item Implements the level-$k$ $\Z(\theta)$ using the $\mathsf{Safe}\Z(\theta)$ construction;
	\item Sends the produced quantum state to the Receiver.
	\end{itemize}	    
  \end{algorithmic}
\end{protocolenv}

While it seems like this protocol only involves the Sender and therefore is always secure, recall that the preparation of the $\ket +$ state and the $\mathsf{Safe}\Z(\theta)$ operation are constructed from Stochastically Compromised Preparations, Gates and Measurements (Resource~\ref{res:scpg}). This Resource allows the Receiver in the protocol above to acquire information and add noise by acting as the Eavesdropper in the resource. We will now show that our constructions allows the Sender to avoid leaking information about $\theta$ and produce a state which is negligibly close to the expected encoded $\ket{+_\theta}$ state even in the presence of a malicious Receiver.

\begin{theorem}
  Protocol~\ref{proto:ft-rsp} $\epsilon$-constructs the Level-$k$ Remote State Preparation (Resource~\ref{res:k-rsp}) from Stochastically Compromised Preparations, Gates, and Measurements (Resource~\ref{res:scpg}) with perfect correctness and soundness error $\epsilon \leq 4 p_0(p_c/p_0)^{2^k}$.
\end{theorem}

\begin{proof}
  First note that the number of locations in the our unencoded circuit is $L=2$: one for the preparation of the $\ket{+}$ state, and another for the application of the $\Z(\theta)$ gate.
  
  To show correctness we must prove that, without adversarial interference, the input-output relation in the ideal and real cases are close. The ideal case corresponds to the Sender and Receiver interacting with the Level-$k$ Remote State Preparation Resource. The real case corresponds to Protocol~\ref{proto:ft-rsp}. In fact, in this situation, there are no errors since the Receiver acts honestly and does not request to access compromised gates. Therefore the protocol produces the expected logical $\ket{+_{\theta}}$ state, and the Resource produce the same state for input $\theta$.

To show the security of Protocol~\ref{proto:ft-rsp} against a malicious Receiver, we make use of the explicit construction of Simulator~\ref{sim:ft-rsp}. We then prove that no Distinguisher can tell apart the malicious Receiver's interactions in (i) the ideal world in which the Simulator has single-query oracle access to the Level-$k$ Remote State Preparation (Resource~\ref{res:k-rsp}), and (ii) the real world in which the honest Sender uses with Stochastically Compromised Preparations, Gates and Measurements (Resource~\ref{res:scpg}).

\begin{simulatorenv}[Simulator for Malicious Receiver in Protocol~\ref{proto:ft-rsp}]\label{sim:ft-rsp}
\item 
  \begin{algorithmic}[0]
    \STATE \textbf{The Simulator:} 
	\begin{itemize}
  	\item Asks Resource~\ref{res:k-rsp} to cheat;
  	\item Samples at random an angle $\theta' \in \Theta$;
  	\item Performs Protocol~\ref{proto:ft-rsp} acting as the Sender with input $\theta'$, with the Receiver as the Eavesdropper;
  	\item Decodes the state obtained at the end of the protocol into a logical qubit and syndrome qubits;
  	\item Sends the syndrome qubits to Resource~\ref{res:k-rsp}, receives in return a state;
  	\item Forwards this state to the Receiver.
	\end{itemize}	    
  \end{algorithmic}
\end{simulatorenv}

  Recall that we have shown above that the level-$k$ preparation of $\ket +$ followed by a $\Z(\theta)$ rotation using the $\mathsf{Safe}\Z(\theta)$ implementation is accurate and private (in the sense of Definition~\ref{def:cor-sec}) with a probability of failure that is bounded above by $2 p_0(p_c/p_0)^{2^k}$. That is, the logical qubit is in the state $\ket{+_{\theta}}$ while the syndrome qubits are in a state that is independent of $\theta$.

  In the simulation, Resource~\ref{res:k-rsp} produces the expected logical state $\ket{+_{\theta}}$ and, because the syndromes produced by the protocol for  $\theta$ or $\theta'$ are indistinguishable as they are directly provided by the distinguisher, the quantum states obtained at the Receiver's interface are indistinguishable. This concludes the proof as the only distinguishing advantage stems from a failure of encoded level-$k$ preparation and rotation to produce the expected logical state together with a syndrome state independent of $\theta$. This happens at most with probability $2 p_0(p_c/p_0)^{2^k}$. The extra factor of two to obtain the construction error is due to the maximum diamond norm. 
\end{proof}
 
\section{Leak- and Fault-Free SDQC Protocols}
\label{sec:composing}
\subsection{Unencoded Composition: Trading Correctness for Security}
\label{sec:issue-correctness}
In this subsection, we develop an SDQC protocol which only deals with leaks, not errors and faults. This covers the case in which the computation is small enough that the noise for the whole computation remains controlled, but the operations on the Verifier's device might leak information due to the physical system used to transmits its encrypted qubits. In this case, we do not need to apply fault-tolerant techniques but still have to suppress these leaks.

\paragraph{A Leak-Tolerant SDQC at the Physicial Level.}
To obtain a leak-tolerant SDQC Protocol, it is possible to compose Protocol~\ref{proto:state_prep} which suppresses the source of leaks with Protocol~\ref{proto:dummyless-sdqc} that performs the secure delegation. This is given in Protocol~\ref{proto:lt-sdqc}. In short, it performs $N$ UBQC computations (Protocol~\ref{proto:ubqc}) on a graph $G = (V,E)$, and each vertex in $V$ requires the Verifier to prepare $\kappa$ leaky $\ket{+_\theta}$ states that are combined by the Prover into a non leaky $\ket {+_\theta}$ state.

\begin{protocolenv}[Leak-Tolerant SDQC Protocol for $\BQP$ Computations]
  \label{proto:lt-sdqc}
  \item 
  \begin{algorithmic}[0]
    
    \STATE \textbf{Public Information:} 
    \begin{itemize}
    \item $G = (V, E, I, O)$, a graph with input and output vertices $I$ and $O$ respectively;
    \item $\mathfrak{W}$, the set of dummyless tests on graph $G$;
    \item $\preceq_G$, a partial order on the set $V$ of vertices;
    \item $N, d, w$, parameters representing the number of runs, the number of computation runs, and the number of tolerated failed tests.
    \item $\kappa$, the number of states used to plug leaks in Protocol~\ref{proto:state_prep}
    \end{itemize}
    
    \STATE \textbf{Verifier's Inputs:} A set of angles $\{\phi(i)\}_{i \in V}$ and a flow $f$ which induces an ordering compatible with $\preceq_G$.
    
    \STATE \textbf{Protocol:}
    \begin{enumerate}
    \item The Verifier samples uniformly at random a subset $C \subset [N]$ of size $d$ representing the runs which will be its desired computation, henceforth called computation runs.
    \item For $k \in [N]$, the Verifier and Prover perform the following:
      \begin{enumerate}
      \item If $k \in C$, the Verifier sets the computation for the run to its desired computation $(\{\phi(i)\}_{i \in V}, f)$. Otherwise, the Verifier samples uniformly at random a test $W$ from the set of dummyless tests $\mathfrak{W}$.
      \item The Verifier and Prover blindly execute the run using the UBQC Protocol~\ref{proto:ubqc}, where each call to the RSP Resource~\ref{res:rsp} is replaced by an execution of Protocol~\ref{proto:state_prep}.
      \item If it is a test, the Verifier computes the XOR of measurement outcomes on the vertices from set $W$. The test accepts if it is $0$ and fails otherwise.
      \end{enumerate}
    \item At the end of all runs, let $x$ be the number of failed tests. If $x \geq w$, the Verifier rejects and outputs $\bot$.
    \item Otherwise, the Verifier accepts the computation. It performs a majority vote on the output results of the computation runs and sets the result as its output.
    \end{enumerate}  
  \end{algorithmic}
\end{protocolenv}

The security of Protocol~\ref{proto:lt-sdqc} is then captured by Theorem~\ref{thm:lt-sdqc}.

\begin{theorem}[Leak-Free Secure Delegation of Quantum Computations]
\label{thm:lt-sdqc}

  Let $d$ be proportional to $N$, and let $c$ the fixed bounded error of the $\BQP$ class of computations.  Let $\epsilon$ be the (constant) lower bound on the probability that a test sampled uniformly from set $\mathfrak{W}$ rejects in the presence of a $\Z$ error.  Let $w$ be the maximum number of test rounds allowed to fail, chosen such that $w < \frac{2c-1}{2c-2}(N-d)(1 - \epsilon)$. Let the authorized leak be defined as $l_{\cptp U} = (G, \mathfrak{W}, \preceq_G)$, and let $\kappa$ be the number of states used to plug leaks in Protocol~\ref{proto:state_prep}.

  Then, Protocol~\ref{proto:lt-sdqc} $\eta_{\mathrm{sec}}(N)$-constructs the Secure Delegated Quantum Computation (Resource~\ref{res:dqc}) from Stochastically Leaky Remote State Preparations (Resource \ref{res:stoca}) in the Abstract Cryptography framework, for $\eta_{\mathrm{sec}}(N)$ negligible in $N$.
  
\end{theorem}

\begin{proof}

The analysis of the correctness and security of this protocol relies on the composition rule from Abstract Cryptography. The correctness and security errors of the composition are the sum of the correctness and security errors of its constituents.

Regarding the correctness, each constituent protocol is perfectly correct so that the resulting protocol is also perfectly correct.

For the security error, it is the sum of the security error for Protocol~\ref{proto:dummyless-sdqc} and the cumulated security errors for each of the $N\times|V|$ call to the the Plugged Rotated Remote State Preparation (Protocol~\ref{proto:state_prep}) as an implementation of the Remote State Preparation (Resource~\ref{res:perf}). That is:
\begin{equation}
  \label{eq:sec_err}
  \eta_{\mathrm{sec}}(N) \leq \eta(N) + N \times |V| \times \frac{7}{8} p_l^\kappa,
\end{equation}
where $\eta(N)$ is negligible in $N$ as per the security of Protocol~\ref{proto:dummyless-sdqc} given in Theorem~\ref{thm:sec-dummyless}. As a result, $\eta_{\mathrm{sec}}(N)$ is negligible in $N$ for $\kappa \in O(N)$.
\end{proof}

\paragraph{Correctness vs.\ Security Trade-off.}
While the parallel repetition of leaky RSP resources can lead to a significant improvement in the security of Protocol~\ref{proto:dummyless-sdqc} when only Stochastic Leak RSP (Resource~\ref{res:stoca}) is available, it poses the risk of quickly deteriorating the correctness of the leak-tolerant protocol in the presence of noise.

To see this, we amend the resources described in Section \ref{subsec:leak_model} to integrate some specific secret-independent noise. This is done by applying a noise channel to the quantum output of the RSP. This corresponds to the case where states produced by the Verifier are sent to the Prover using a noisy quantum channel.  To ease the discussion, we consider below the special case of a depolarising channel with probability $p_\mathit{noise}$ of producing a completely mixed state. We call such resources Noisy Leaky RSPs.

We now show that $\kappa$ such RSPs can be used to construct a non-leaky RSP with a noise parameter equal to $1 - (1-p_\mathit{noise})^\kappa$.  In Protocol~\ref{proto:state_prep}, each of the  $\ket{+_\theta}$ states one-time-pads the final state. Therefore the noise of a completely mixed state accumulates and yields a completely mixed, and therefore useless, target state. For the depolarising noise described above, the probability to obtain a valid final state hence equals the probability that all noisy RSP resources return perfect states:
\begin{align*}
  (1-p_\mathit{noise})^\kappa,
\end{align*}
which decreases exponentially in the number of states used in Protocol \ref{proto:state_prep}.  The Simulator for the noisy version of the construction works exactly as the one described Section~\ref{sec:rotated-rsp}. As a consequence of Theorem~\ref{thm:boost} and Lemma~\ref{lem:stoc-to-perf}, we know that this construction therefore implements a noisy non-leaky RSP resource $\epsilon_\text{sec}$-securely, where $\epsilon_\text{sec} = p_\mathit{leak}^\kappa$.

While this bound is also exponential in $\kappa$, it admits a different base than the noise bound above, which we can exploit to our advantage. In particular, for $p_\mathit{noise}, p_\mathit{leak} \in (0, 1/2)$ the gap between the resulting noise parameter and $\epsilon_\text{sec}$ grows exponentially in $\kappa$ as well.  For a fixed noise parameter bound $c$ we get:
\begin{align*}
  (1-p_\mathit{noise})^\kappa > c,
\end{align*}
which immediately implies the following upper bound on $\kappa$:
\begin{align*}
  \kappa < \frac{\log(c)}{\log(1-p_\mathit{noise})},
\end{align*}
This, in turn, implies a lower bound on the security that is achieved with this bounded number of repetitions:
\begin{align*}
  \epsilon_\text{sec} = p_\mathit{leak}^\kappa > p_\mathit{leak}^{\frac{\log(c)}{\log(1-p_\mathit{noise})}} = c^{1/\log_{p_\mathit{leak}}(1-p_\mathit{noise})}.
\end{align*}
In effect, this means there exists a trade-off between secret-dependent (leaky) noise and secret-independent noise in the RSP. Using parallel repetitions, it is possible to suppress the secret-dependency of the noise while accumulating more secret-independent noise in the final distilled state. 
 
\subsection{Encoded Composition: Trading Complexity for Robustness}
\label{sec:ltft-sdqc}
The clearly identifiable drawback of the previous subsection is that for any fixed per-gate noise level, and for any security error, there will always exist a maximum computation size such that beyond this size, computations cannot be made secure within the chosen security error. It does come however with an advantage. All the complexity of the protocol is taken care of by the Prover, making it economically advantageous in scenarios where multiple clients delegate their computation to a single quantum machine.

In this subsection, we explore the complementary regime where the Verifier's is more capable but manages to delegate arbitrary long computations in spite of noise on its side and against an arbitrary Prover. Indeed, we will require the Verifier to be able to prepare and perform rotations for single qubit encoded states contained in registers whose size increases polylogarithmically with the length of the computation and with the inverse of the desired security error. In this sense it will provide a scalable and efficient solution to the problem of secure delegation of quantum computation between a noisy Verifier and an arbitrary dishonest Prover.

The protocol follows a similar template as the one exhibited in the previous subsection: the Verifier relies on its ability to perform a RSP and then delegates the computation using Protocol~\ref{proto:dummyless-sdqc}. The specificity here is the use of Protocol~\ref{proto:ft-rsp} to construct a Level-$k$ RSP (Resource~\ref{res:k-rsp}) that produces a logically encoded $\ket{+_{\theta}}$ state. Then the encoded qubits are directly used by the Prover according to the instructions given by the Verifier, at the logical level. Note that the Verifier does not need to assist the Prover in performing the error correction as the secrets are only used to protect the logical information and not physical qubits. In this situation, the Prover is thus fully autonomous.

We thus define the following protocol, for a computation given as a measurement patter on a graph $G=(V,E)$:
\begin{protocolenv}[Secure Delegation of Fault-Tolerance Quantum Computation]
\label{proto:ft-sdqc}
\item
\begin{algorithmic}[0]
\STATE \textbf{Public Information:} 
\begin{itemize}
\item $G = (V, E, I, O)$, a graph with input and output vertices $I$ and $O$ respectively;
\item $\mathfrak{W}$, the set of dummyless tests on graph $G$;
\item $\preceq_G$, a partial order on the set $V$ of vertices;
\item $N, d, w$, parameters representing the number of runs, the number of computation runs, and the number of tolerated failed tests;
\item $k$, the concatenation level used to encode logical information.
\end{itemize}
    
\STATE \textbf{Verifier's Inputs:} A set of angles $\{\phi(i)\}_{i \in V}$ and a flow $f$ which induces an ordering compatible with $\preceq_G$.

\STATE \textbf{Protocol:}
The Verifier and Prover execute Protocol~\ref{proto:dummyless-sdqc} for the intended pattern on graph $G=(V,E)$ in the following way:
\begin{itemize}
\item For each call to the RSP, they run Protocol~\ref{proto:ft-rsp};
\item For each gate or measurement instructed by the Verifier, the Prover implements the corresponding level-$k$ simulation, considering for measurements that  the outcome is that for the measurement of the corresponding logical qubit.
\end{itemize}
\end{algorithmic}
\end{protocolenv}

We can now state the main theorem.

\begin{theorem}[Threshold Theorem for Secure Delegation of Fault-Tolerant Quantum Computations]
  Let $d$ be proportional to $N$, and let $c$ be the fixed bounded error of the $\BQP$ class of computations.  Let $\epsilon$ be the (constant) lower bound on the probability that a test sampled uniformly from set $\mathfrak{W}$ rejects in the presence of a $\Z$ error.  Let $w$ be the maximum number of test rounds allowed to fail, chosen such that $w < \frac{2c-1}{2c-2}(N-d)(1 - \epsilon)$.  Let the leak be defined as $l_{\cptp U} = (G, \mathfrak{W}, \preceq_G)$, and let the Verifier use a register of size $15^{k}$ corresponding to $k$ levels of concatenation of the quantum $[[15,1,3]]$ Reed-Muller code with $15^k \in O(Nlog(|V|))$.
  Let $p_c$ be the probability that a location is compromised, and $p_0$ be the threshold probability from Lemma~\ref{lem:rare_bad_recs} and Theorem~\ref{thm:safe_threshold}

  Then, for $p_c < p_0$, Protocol~\ref{proto:ft-sdqc} $\eta(N)$-constructs the Secure Delegated Quantum Computation (Resource~\ref{res:dqc}) from Stochastically Compromised Preparations, Gates and Measurements (Resource \ref{res:scpg}) in the Abstract Cryptography framework for $\eta(N)$ negligible in $N$.
\end{theorem}

\begin{proof}
  The proof relies on composing protocols using Abstract Cryptography. As outlined in Protocol~\ref{proto:ft-sdqc}, it relies onto Protocol~\ref{proto:dummyless-sdqc} calling Protocol~\ref{proto:ft-rsp} in place of the original RSP (Resource~\ref{res:k-rsp}). As a result, the error in constructing the SDQC resource is the sum of the error pertaining to Protocol~\ref{proto:dummyless-sdqc} and of the error pertaining to Protocol~\ref{proto:ft-rsp} multiplied by the number of times it is called, i.e. $N\times|V|$. This gives the following bound on the construction error:
  \begin{equation}
  \epsilon_{\mathrm{SDFTQC}} \leq \eta(N) + N\times|V|\times 4 p_0 \left( \frac{p_{\mathrm{c}}}{p_0} \right)^{2^k}.
  \end{equation}
Clearly, this can be made as small as desired by a modest increase of $k$. To get a more concrete result, we can further request that both terms contributing to $\epsilon_{\mathrm{SDFTQC}}$ are of comparable magnitude, and see how the size of the register that the Verifier needs to manipulate is impacted.

To ensure that $\epsilon_{\mathrm{SDFTQC}} \leq 2 \eta(N)$ we must set $k$ according to
\begin{equation}
N\times|V|\times 4 p_0 \left( \frac{p_{\mathrm{c}}}{p_0} \right)^{2^k} \leq \eta(N)  
\end{equation}
That is
\begin{equation}
  2^k \leq \left( \frac{\log( N \times |V| \times 4 p_0 / \eta(N))}{\log(p_0/p_{\mathrm{c}})}  \right)
\end{equation}
A closer inspection into the specific form of $\eta(N)$ reveals that it is in $O(\exp(-N))$ so that we have a register size equal to $15^k$ that is in $\tilde O(N)$.
\end{proof}

\begin{corollary}[Efficiency and Scalability]
  Protocol~\ref{proto:ft-sdqc} is both efficient and scalable because
\begin{itemize}
\item \textit{(Efficiency)}: Given a computation class defined through a fixed graph $G=(V,E)$, reaching a target construction error imposes a logarithmic overhead in the number of repetitions and the size of the single logical qubit register that the Verifier needs to manipulate;
  \item \textit{(Scalability)}: For a fixed construction error, increasing the size of the graph induces only a logarithmic increase in the size of the register. 
\end{itemize}
\end{corollary}

\section{Discussion and Future Work}
\label{sec:concl}
\paragraph{Summary and Significance of Results.}
  
In this paper we have presented a novel model for taking account noisy operations on the Verifier's side and shown how it impacts the security of its operations in the context of Secure Fault-Tolerant Delegated Quantum Computations. Dealing with noise in this context forced us to consider how information about the state is leaked due to the imperfections. We have therefore redesigned staple elements of fault-tolerant quantum computations. We have constructed first a protocol which does not increase the complexity of Verifier operations, but is only able to deal with leaky operations. As soon as noise is added to this protocol, a trade-off between security and correctness appears. Our second protocol is able to deal both with noisy and leaky operations, however the Verifier now needs to implement its operations in a fault-tolerant manner. Thankfully, the circuit applied by the Verifier is only of depth $2$, meaning that it is not required to perform arbitrary computations. Our protocols therefore settle the long-standing open-question of efficient verification of quantum computations in the fault-tolerant regime.

Beyond this use-case, we can use our protocol as a template for designing efficient benchmarking tools to measure the performances of future early fault-tolerant computers. Passing the tests from our SDQC protocol for a given graph would also guarantee that any $\BQP$ computations using the same graph are perfectly implementable on the machine. Furthermore, information about which tests fail or pass can be used to determine noise parameters given appropriate noise models.

Our techniques for removing secret-dependent noise, whether it be by state distillation or gate compilation, could be applied to other cryptographic protocols so long as they require a step involving RSP in the $\X - \Y$ plane. This would allow them to also mitigate the same type of leakage as the one we consider in the current paper.

While it is possible to run these protocols with the Verifier and the Prover being in separate facilities, we can also use it on a single machine if we are willing to add physical assumptions. For example, they can be run so long as the Verifier section is implemented in a trusted environment, or if we can bound the amount of leakage from the section implementing the Verifier to the section of the machine performing Prover operations.

\paragraph{Open Questions and Future Work.}

The first pair of open questions which exceed the scope of this paper are (i) how to deal with leaks that are modelled as always leaking only a small part of the secret parameters, and (ii) how to deal with noise models in which a CPTP map $\mathcal{F}_{\theta}$ which depends on $\theta$ is always applied to the state $\ketbra{+_\theta}$ but with the guarantee that it is $\epsilon$-close to the identity in diamond norm, for a small value of $\epsilon$.

One issue with the dummyless SDQC protocol used in this work is that it can only handle computations with classical outputs since all qubits need to be measured. Therefore is it well suited for $\BQP$ computations but does not cover the case where a quantum output is desirable. On the other hand, the fact that multiple runs need to be repeated in our SDQC protocol makes it impossible to handle quantum inputs, of which the Verifier may only have one copy. Other SDQC protocols can handle quantum inputs and outputs but require the Verifier to produce both $\ket{+_\theta}$ states but also computational basis states. In order to use these protocol with our model of leaks, errors and faults, novel RSP techniques would need to be designed, in particular a secure $1-\mathsf{Ga}$ similar to the $1-\mathsf{Safe}\Z(\theta)$ we describe.

The Verifier in our protocol does not requires a full fault-tolerant quantum computer, but still needs to be able to operate on possibly large encoded states. Ideally, we would like it to be as close to purely classical as possible. While the question of a fully classical Verifier with statistical security is still open even in the non fault-tolerant case, there have been attempts at FT SDQC protocols in which the Verifier only performs individual single-qubit physical operations~\cite{TFMI17fault}. Unfortunately, we show in Appendix~\ref{app:attack-tomoyuki} that this protocol is broken. Therefore, the question of the existence of FT SDQC protocols with single-qubit physical Verifier operations remains open.

\ifsodasubmission
\else
	\paragraph{Acknowledgements.}
	The authors want to thank Elham Kashefi for fruitful discussions and for keeping them motivated to work on this topic.
	TK was supported by a Leverhulme Early Career Fellowship.
	DL acknowledges support from the French ANR Project ANR-21-CE47-0014 (SecNISQ), and from the Quantum Advantage Pathfinder (QAP) research program within the UK's National Quantum Computing Center (NQCC).
	HO acknowledges the support from France 2030 under the French National Research Agency award number “ANR-22-PNCQ-0002” (HQI).
	This work has been co-funded by the European Commission as part of the EIC accelerator program under the grant agreement 190188855 for SEPOQC project, and by the Horizon-CL4 program under the grant agreement 101135288 for EPIQUE project.
\fi
 
\newpage

\bibliographystyle{alpha} 
\bibliography{leaky-rsp-stripped.bbl}

\newpage

\appendix
    
\section{Formal UBQC Protocol}\label{app:ubqc}
\begin{protocolenv}[UBQC Protocol for Classical Inputs and Outputs]
  \label{proto:ubqc}
  \item
  \begin{algorithmic}[0]

    \STATE \textbf{Verifier's Inputs:} A measurement pattern $(G, I, O, \{\phi(i)\}_{i\in V}, f)$.

    \STATE \textbf{Protocol:}
    \begin{enumerate}
    \item The Verifier sends the graph's description $G = (V, E)$ and a measurement order $\preceq$ compatible with the flow $f$ to the Prover;
    \item For all $i \in V$, the Verifier chooses a random $\theta(i) \in \Theta$, performs a call to the RSP Resource~\ref{res:rsp}, which prepares $\ket{+_{\theta(i)}}$ and sends it to the Prover.
    \item The Prover applies a $\CZ$ gate between qubits $i$ and $j$ if $(i,j) \in E$ is an edge of $G$;
    \item For all $i \in V$, in the order specified by $\preceq$, the Verifier samples a random bit $r(i)$, computes the measurement angle $\delta(i)$ and sends it to the Prover, receiving in return the corresponding measurement outcome $b(i)$:
      \begin{align*}
      	s(j) & = b(j) \oplus r(j), \\
        s_X(i) & = \bigoplus_{j \in S_X(i)} s(j), \quad
        s_Z(i) = \bigoplus_{j \in S_Z(i)} s(j), \\
        \delta(i) & = (-1)^{s'_X(i)}\phi(i) + \theta(i) + (s_Z(i) + r(i)) \pi ,
      \end{align*}
    \item The Verifier sets $\{s(i)\}_{i \in O}$ as its output.
    \end{enumerate}
  \end{algorithmic}
\end{protocolenv}

\section{Attack on Previous Fault-Tolerant SDQC Protocol}
\label{app:attack-tomoyuki}
We present here an attack on the only previously presented fault-tolerant SDQC protocol from \cite{TFMI17fault}. Their claim is that fault-tolerant SDQC is possible with a Verifier that performs only physical $\X$ and $\Z$ measurements and classical post-processing. To do so, they rely on the protocol from \cite{FK17unconditionally} to perform the SDQC computation, only replacing the Verifier having to send the logical state with a first step of logical Remote State Preparation via a dedicated gadget to prepare the encoded states on the Prover's side. The code used is the CSS code. We first describe this gadget in Protocol~\ref{proto:tomo-gadget}, presented in Section 2 from \cite{TFMI17fault}, and then show how to break the protocol's security guarantees.

\begin{protocolenv}[Logical Remote State Preparation Gadget]
  \label{proto:tomo-gadget}
\begin{algorithmic} [0]
    \STATE \begin{enumerate}
    	\item The Sender samples three sets of five bits $\{c_i, a_i, r_i\}_{1\leq i \leq 5}$ from $\bin^{15}$.\footnotemark
    	\item The Receiver prepares 5 copies of the logical EPR state $\frac{1}{\sqrt{2}}(\ket{00} + \ket{11})$ and sends of of each EPR pair to the Sender.
    	\item For $1 \leq i \leq 5$:
    	\begin{enumerate}
    		\item The Sender:
    		\begin{itemize}
    			\item If $c_i = 0$, it measure the $i$\textsuperscript{th} received state in the logical $\Z$ basis. It corrects the measurements to obtain the corrected outcome $o_i$ and sends back the correction $\X^{a_i \oplus o_i} \Z^{r_i}$ to the Receiver.
    			\item If $c_i = 1$, it measure the $i$\textsuperscript{th} received state in the logical $\X$ basis. It corrects the measurements to obtain the corrected outcome $o_i$ and sends back the correction $\X^{r_i} \Z^{a_i \oplus o_i}$ to the Receiver.
    		\end{itemize}
    		\item The Receiver applies the correction to its half of the $i$\textsuperscript{th} EPR pair.
    	\end{enumerate}
    	\item The Receiver applies the circuit presented in Figure~\ref{fig:tomo-circ}, and keeps the unmeasured logical qubit as its output if the measurement outcomes are all equal to $0$ and aborts otherwise.
    	\item If there is no abort, the Sender keeps as its output the classical description of the state of the unmeasured logical state, given in Table~\ref{tab:tomo-out}
    \end{enumerate}
  \end{algorithmic}
\end{protocolenv}

\footnotetext{The probabilities are tweaked so that each the fraction of prepared states in each basis after post-selection corresponds to the requirements of the SDQC protocol from~\cite{FK17unconditionally}. However, the exact probability distribution is irrelevant for our attack.}

\begin{figure}[htp]
\centering
$\Qcircuit @C=1em @R=.7em @!R {
\lstick{\rho_1} & \gate{\Z(\pi/4)} & \control \qw  & \qw           & \gate{\Ha} & \qw & \meter \\
\lstick{\rho_2} & \qw              & \ctrl{-1}     & \control \qw  & \gate{\Ha} & \qw & \meter \\
\lstick{\rho_3} & \qw              & \control \qw  & \ctrl{-1}     & \qw        & \qw & \qw & \rstick{\ket{B}}\\
\lstick{\rho_4} & \qw              & \ctrl{-1}     & \control \qw  & \gate{\Ha} & \qw & \meter \\
\lstick{\rho_5} & \gate{\Z(\pi/2)} & \qw           & \ctrl{-1}     & \gate{\Ha} & \qw & \meter
}$
\caption{The quantum circuit used by the Receiver in the gadget after having corrected all five half-EPR pairs, corresponding to the states $\rho_i$.}
\label{fig:tomo-circ}
\end{figure}
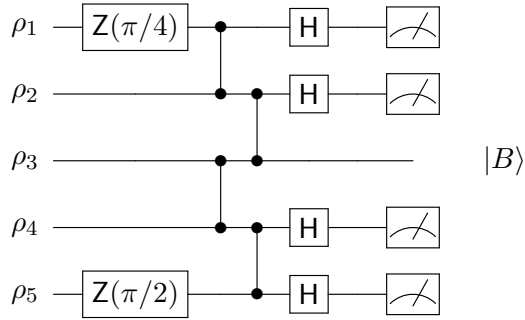

\begin{table}[htp]
\centering
\begin{tabular}{|c|c|c|} \hline
& $(c_1,c_2,c_3,c_4,c_5)$ & $\ket{B}$ \\ \hline\hline
(1) & $(0/1,0/1,0,0/1,0/1)$ & $\X^{a_3}\ket{0}$ \\ \hline
(2) & $(0,1,1,0/1,0/1)$ & $\X^{a_1\oplus a_2}\ket{0}$ \\ \hline
(3) & $(0/1,0,1,1,0)$ & $\X^{a_4\oplus a_5}\ket{0}$ \\ \hline
(4) & $(1,1,1,1,0)$ & $\X^{a_4\oplus a_5}\ket{0}$ \\ \hline
(5) & $(0/1,0,1,0,0/1)$ & $\Z^{a_2\oplus a_3\oplus a_4}\ket{+}$ \\ \hline
(6) & $(0/1,0,1,1,1)$ & $\Z^{a_2\oplus a_3\oplus a_4\oplus a_5}\ket{+_{\frac{\pi}{2}}}$ \\ \hline
(7) & $(1,1,1,0,0/1)$ & $\X^{a_2}\Z^{a_1\oplus a_3\oplus a_4}\ket{+_{\frac{\pi}{4}}}$ \\ \hline
(8) & $(1,1,1,1,1)$ & $\X^{a_2}\Z^{a_1\oplus a_2\oplus a_3\oplus a_4\oplus a_5}\ket{+_{\frac{3\pi}{4}}}$ \\ \hline
\end{tabular}
\caption{The explicit form of $\ket{B}$ when post-selected on all outputs being $0$ in the circuit from Figure~\ref{fig:tomo-circ}. Here, $0/1$ means that either $0$ or $1$ will yield the desired outcome.}
\label{tab:tomo-out}
\end{table}

It can easily be checked that the protocol is correct, in the sense that the Receiver indeed has in its possession at the end of an honest execution the state corresponding to the classical description obtained by the Sender. This relation is given by Table~\ref{tab:tomo-out}. Furthermore, the Sender only measures single physical qubits since the $\X$ and $\Z$ logical measurements are transversal for this code.

One crucial aspect of the security proof of SDQC protocols is that the deviation of the adversary does not depend on the classical description of the states sent by the Verifier. In fact the authors of~\cite{TFMI17fault} acknowledge this in their security proof (cf. Remark 1 from Section 6 of their paper) and claim that the deviation indeed does not depend on the classical descriptions $c_i, a_i$.

Such a dependency can in fact completely break a protocol, as demonstrated in~\cite{KKLM23asymmetric}. Indeed they show that, if the Prover can apply a logical $\X$ operation on the state $\ket{0}, ket{1}$ selectively without impacting the states $\ket{+_\theta}$, then such a Prover is capable of making the Verifier accept an incorrect outcome with probability $1$. We show that a Prover is capable of mounting such an attack on the Logical RSP protocol of~\cite{TFMI17fault} and that therefore their protocol is not verifiable.

The attack corresponds to the Prover applying the circuit from Figure~\ref{fig:tomo-attack} instead of the one presented in Figure~\ref{fig:tomo-circ}.

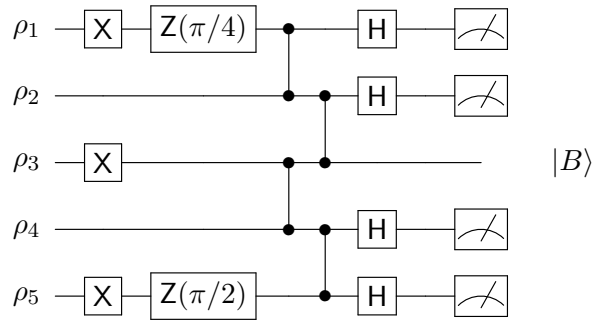
\begin{figure}[htp]
\centering
$\Qcircuit @C=1em @R=.7em @!R {
\lstick{\rho_1} & \gate{\X} & \gate{\Z(\pi/4)} & \control \qw  & \qw           & \gate{\Ha} & \qw & \meter \\
\lstick{\rho_2} & \qw       & \qw              & \ctrl{-1}     & \control \qw  & \gate{\Ha} & \qw & \meter \\
\lstick{\rho_3} & \gate{\X} & \qw              & \control \qw  & \ctrl{-1}     & \qw        & \qw & \qw & \rstick{\ket{B}}\\
\lstick{\rho_4} & \qw       & \qw              & \ctrl{-1}     & \control \qw  & \gate{\Ha} & \qw & \meter \\
\lstick{\rho_5} & \gate{\X} & \gate{\Z(\pi/2)} & \qw           & \ctrl{-1}     & \gate{\Ha} & \qw & \meter
}$
\caption{The quantum circuit used by the malicious Receiver.}
\label{fig:tomo-attack}
\end{figure}

We can now look at the effect that this attack has on the prepared state. Notice that the inputs to the circuit above are either in the logical $\ket{0}$, $\ket{1}$, $\ket{+}$ or $\ket{-}$ state. Therefore a logical $\X$ from the attack only has an effect in the former two cases, i.e.~if $c_i = 0$ for $i \in \{1, 3, 5\}$. Furthermore, in that case, since the input state $i$ is of the form $\Ha^{c_i}\X^{a_i}\ket{0}$, applying $\X$ is equivalent to flipping the value of $a_i$.

We can then look at the output of the circuit given by table~\ref{tab:tomo-out} and see the effect of flipping $a_i$ for $i \in \{1, 3, 5\}$ if $c_i = 0$. We see that $\X$ on qubit $1$ flips the outcome in case (2), $\X$ on qubit $3$ flips the outcomes in case (1) and $\X$ on qubit $5$ flips the outcome in cases (3) and (4). These cases correspond to all of which yield an outcome in the $\Z$ logical basis on no other.

Therefore our attack selectively flips the produced qubit conditioned on the fact that it is in the logical computational basis. This invalidates the claim in the security proof of~\cite{TFMI17fault} that any attack of the adversary is independent of the state sent by the Verifier. We note that at no point does this require us to violate the blindness of the protocol to do so, nor know what the values of $c_i, a_i, r_i, o_i$ are. However, knowing how each position in the circuit influences the state that is produced is sufficient to break the scheme.

The rest of the protocol in~\cite{TFMI17fault} does not specify sufficiently how the states produced by the gadget are used. Indeed, in the original protocol from~\cite{FK17unconditionally}, the Verifier chooses for each qubit which state is sends, while the gadget as specified only produces these states at random. It is specified only that the gadget is repeated until enough states are generated of each kind. We assume that the Verifier then request of the Prover that only a subset of states be used and a random permutation applied to them. However, this is rather wasteful since it generates more states than what the SDQC protocol of~\cite{FK17unconditionally} requires.

Note that the Verifier can also choose the values of $c_i$ in the gadget so that the basis of the state is fixed: it is wishes to send a computational basis state, it samples the $c_i$ uniformly at random from the set that yield the first four cases of Table~\ref{tab:tomo-out}, otherwise, it samples from the complement to send a random $\ket{+_\theta}$ state. In that case, a similar attack as the one presented in~\cite{KKLM23asymmetric} can be mounted, which corrupts the outcome of the Verifier's computation with probability $1$ while remaining completely undetected.
 
\end{document}